\documentclass[11pt]{article} 
\let\savedbbchi\bbchi
\usepackage{graphicx}
\usepackage{color}
\usepackage{amsmath,amsthm,amssymb}
\usepackage{amsfonts}
\usepackage[normal]{subfigure}
\usepackage{a4wide}
\usepackage{enumitem}
\usepackage{authblk}

\let\bbchi\savedbbchi

\newcommand{\ie}{{\it i.e.}}
\newcommand{\bvm}{\mbox{\sf bvM}_{k, \theta} }
\newcommand{\bvmm}{\mbox{\sf bvM}}
\newcommand{\bvmone}{\mbox{\sf bvM}_{k, e_1} }
\newcommand{\bvmtwo}{\mbox{\sf bvM}_{k, e_2} }
\newcommand{\ep}{\varepsilon}
\usepackage{marvosym}

\newsavebox{\astrutbox}
\sbox{\astrutbox}{\rule[-5pt]{0pt}{20pt}}

\newtheorem{theorem}{Theorem}[section]
\newtheorem{lemma}{Lemma}

\title{Direction-Dependent Turning Leads to\\
Anisotropic Diffusion and Persistence}

\author{Nadia Loy, \footnote{DISMA, Politecnico di Torino, Torino, Italy, \texttt{nadia.loy@polito.it}} 
        Thomas Hillen, \footnote{University of Alberta, \texttt{thillen@ualberta.ca}} 
        Kevin J. Painter\footnote{DIST, Politecnico di Torino, Torino, Italy, \texttt{kevin.painter@polito.it}} 
}

\date{\today}

\begin{document}

\label{firstpage}
\maketitle

\begin{abstract}
Cells and organisms follow aligned structures in their environment, a process that can generate persistent migration paths. Kinetic transport equations are a popular modelling tool for describing biological movements at the mesoscopic level, yet their formulations usually assume a constant turning rate. Here we relax this simplification, extending to include a turning rate that varies according to the anisotropy of a heterogeneous environment. We extend known methods of parabolic and hyperbolic scaling and apply the results to cell movement on micro-patterned domains. We show that inclusion of orientation dependence in the turning rate can lead to persistence of motion in an otherwise fully symmetric environment, and generate enhanced diffusion in structured domains.   
\end{abstract}


\textbf{Subject class (MSC 2020)}
35Q92 (Primary); 92C17, 46N60, 35Q20

\section{Introduction}

Movement of cells through tissues is critical during both healthy and pathological processes. Embryonic development relies on cells migrating from origin to final tissue destination, repair processes necessitate movement of fibroblasts and macrophages into the wound site, and migration of cancerous cells, unhappily, leads to tumour invasion and metastasis dissemination. Consequently, there is clear reason to understand the factors that guide cells with one such process, {\it contact guidance}, defining the movement of cells along linear/aligned tissue features, for example blood vessels, white matter brain fibres, or the collagen fibres of connective tissue.    

The influence of contact guidance on cell movement has been considered via a variety of mathematical approaches \cite{Dickinson_Tranquillo.93,Dickinson,Scianna_Preziosi.13}, with kinetic transport equations proving particularly popular \cite{Hillen.05,Painter.09,Othmer_Hillen.00,Chalub_Markowich_Perthame_Schmeiser.04}. Transport equations account for the microscopic features of movement, describing a migration path according to its statistical properties (turning rate, movement speed and movement direction), and such models for contact guidance have been successfully applied to, for example, glioma invasion \cite{Painter2013,Hillen_Painter.13,Swan.17}. Yet these studies have been simplified through taking turning rates to be independent of orientation, whereas 
experiments indicate considerably more complexity (e.g. \cite{riching2014,Ray2017,ray2018}.

Here we extend the known theory of scaling limits for transport equations in biology to cases where the turning rate is direction-dependent. In the physical context, our model can be seen as a non-homogeneous linear Boltzmann equation  with a micro-reversible process in which the \textit{cross-section} is factorised into a \textit{turning kernel} and a direction-dependent \textit{turning rate} \cite{Petterson, DegondGoudonPoupaud}. The direction dependence of the turning rate augmentation presents mathematical challenges, requiring reflection on how the involved function spaces should be modified for a Fredholm alternative argument to be constructed. We obtain expressions for the macroscopic diffusion and advection that incorporates the influence of sophisticated turning rate choices. We apply the model to the movement data of cells on oriented microfabricated surfaces generated by Doyle et al. \cite{doyle2009}, showing that strong alignment can lead to persistence of movement and, at a macroscopic level, enhanced diffusion. 

The outline of this paper is as follows. We use the remainder of this introduction to provide background on contact guidance and detail experimental investigations of cell movements on micro-pattern domains. We also review pertinent modelling literature, particularly using kinetic transport equations. In Section \ref{model_dir_dep}, we introduce the model, explain the basic assumptions, and introduce statistical meaningful quantities such as mean velocity, runtimes along fibres, directional variance, and persistence.  In Section \ref{mod-turn_macro} we consider two scaling limits, the {\it parabolic limit} (Theorem \ref{t:paralim})  and the {\it hyperbolic limit}. In particular, we generalise the technique proposed in \cite{Hillen.05} to the case of a direction-dependent turning rate and obtain a macroscopic diffusive limit with a distinctive structure. Section \ref{mod_turn_diff_eq} is used to discuss pertinent special cases, with Section \ref{mod_turn_num_tests} illustrating how the new dynamics can result from a direction-dependent turning rate. Extending the analysis to a particularly relevant form,
Section \ref{s:Doyle} is used to demonstrate the utility of the model for describing cell migration paths on microfabricated surfaces. We close with a discussion in Section \ref{s:discussion}.

\subsection{Background}

The extracellular matrix (ECM) is a fundamental ingredient of connective tissues and constitutes the major non-cellular component of tissues and organs. Cell migration through ECM can occur individually or collectively, with individual further classified into amoeboid and mesenchymal forms \cite{boekhorst2016}. Mesenchymal migration is typically slower, with a cell secreting degrading enzymes (e.g. MMPs) that create space for movement. Thus, mesenchymal migration can significantly alter the local ECM structure. Amoeboid migration is often faster, with frequent turns and shape changes allowing a cell to squeeze through matrix gaps; contacts are fleeting, leading to moderate and transient changes to the ECM architecture. Cells may switch between migration modes, for example in response to the biomechanical resistance of the ECM. This {\em mesenchymal-amoeboid transition} \cite{Wolf_Friedl.03} may potentially optimise tumour invasion \cite{Tum_Inv_Opt} in complex heterogeneous micro-environments \cite{Am_Mes}.  

Regardless of migration type, the architecture of the ECM is a major determinant of movement. ECM is formed
from various proteins, with collagen often the principal constituent \cite{alberts2014}. Individual collagen 
proteins are organised into cable-like fibres, collectively creating a network. Adhesive attachments between cells and matrix-binding sites anchor the cell and provide the focal points for exerting the forces needed for forward propulsion. Consequently, by protruding and pulling itself along fibres, cells follow the local topology of the matrix ({\em contact guidance}) \cite{dunn1976}. The mesh formed from collagen therefore offers an example of a {\em bidirectional anisotropic} network, bidirectional in the sense that movement preferentially follows fibres but no specific direction is favoured. Anisotropic bidirectional tissues extend to other environments, a particular relevant example being the brain's white matter. Here it is the long and bundled neuronal axons that generate the network and its arrangement is believed to be a key determinant in the anisotropic invasion of gliomas \cite{giese1996a,giese1996b}.

The question of how an anisotropic environment influences cell migration is highly suited to modelling and 
various approaches have been developed. Agent-based models that incorporate contact guidance include lattice-free 
particle approaches (e.g. \cite{dallon99,mcdougall06,schluter12}) and those based on the Cellular Potts Model (e.g. \cite{Scianna_Preziosi.13.2,Scianna_Preziosi.13,Scianna_Preziosi.14}) and other automata (e.g. \cite{Am_Mes}); the individual-level description is clearly advantageous for incorporating microscopic structure. Continuous models, though, have also been developed, for example the anisotropic biphasic theory (ABT) developed in \cite{Dickinson_Tranquillo.93} and transport equations studied in \cite{Dickinson}. A transport equation developed in \cite{Hillen.05} describes the contact-guided migration of mesenchymal (and amoeboid) cells in evolving anisotropic networks of unidirectional or bidirectional type, with this model extended and subjected to numerical exploration in \cite{Painter.09}. Transport models have a `stepping-stone' nature, lying at a point between an individual and macroscopic model: they sit at a mesoscopic level,
describing the statistical distribution of the individual microscopic velocities and positions through density distribution functions.
Subsequent up-scaling 
can generate a fully macroscopic model, typically of drift-diffusion nature, and capable of capturing 
movement at a large-tissue level. In the study of \cite{Hillen.05} the author employed such scaling techniques to 
recover macroscopic limits. 

\begin{figure}[t!]
    \centering
    \includegraphics[width=0.66\textwidth]{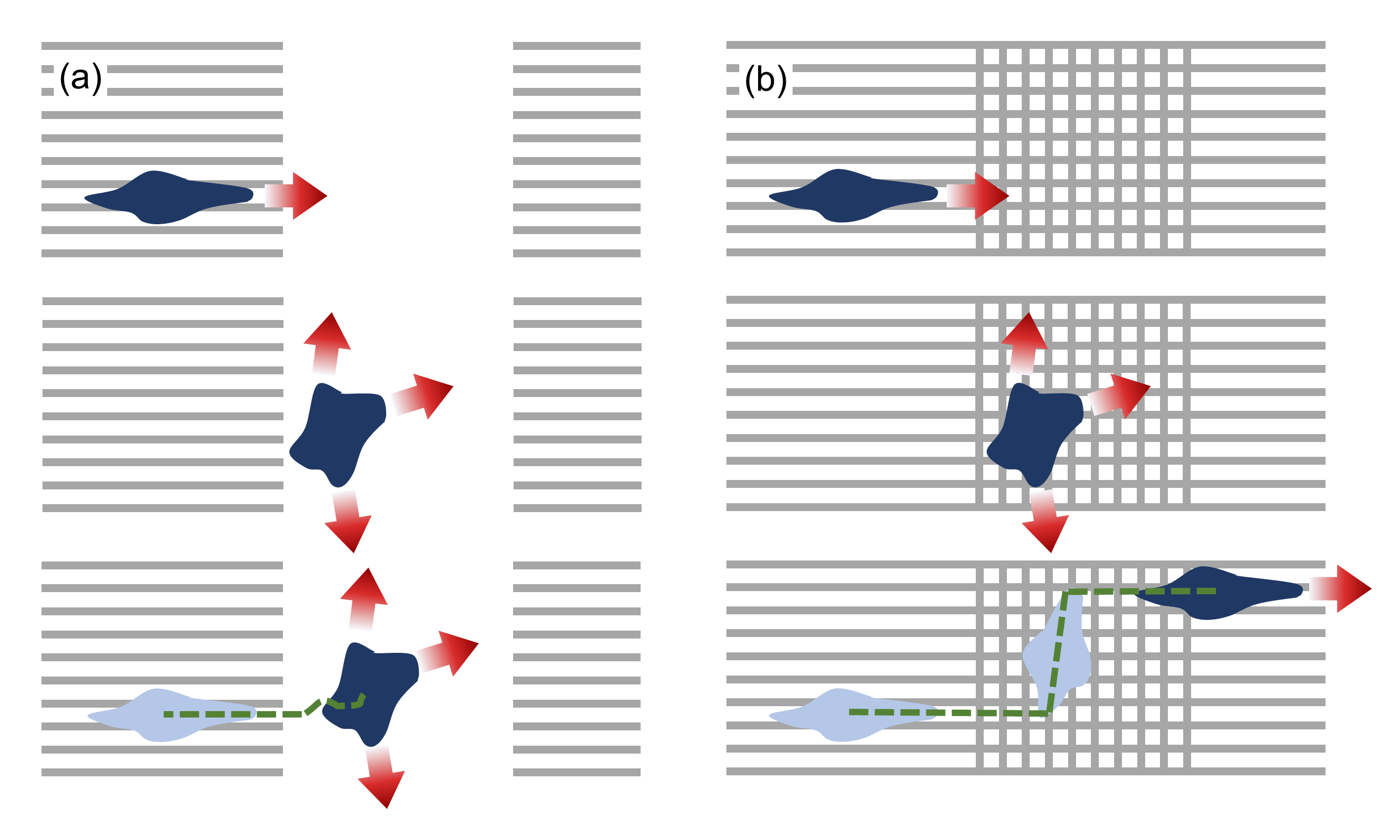}
    \caption{Micropatterning surfaces with fibronectin stripes allow construction of controlled anisotropic environments, for example in the above quasi one-dimensional arrangement of parallel stripes/fibres interrupted by either (a) a completely isotropic region, or (b) a region of criss-crossing perpendicular stripes, see \cite{doyle2009}. For a cell whose turning frequency drops when migrating in a direction of dominating orientations decreases, two very distinct behaviours are conceivable as a cell crosses from the quasi-1D to 2D region.} \label{doyleschema}
\end{figure}

The transport equation in \cite{Hillen.05} is predicated on an underlying stochastic velocity-jump model of migration \cite{Alt.88}, i.e. fixed-velocity runs interspersed with velocity changes. The transitional probability for switching velocities (from a pre-reorientation to post-reorientation velocity) can be decomposed into two elements: a {\em turning rate} function that dictates the rate at which switches occur, and a {\em turning kernel} that describes the selection of the new velocity/direction. The latter was taken to be an angular distribution (potentially space and time varying) that encodes the oriented ECM network structure. Thus, contact-guided migration was included through an increased likelihood of a cell choosing the dominating local fibre orientation. Consequently, the model captures the anisotropic spread of a population in an aligned bidirectional network and simulations in \cite{Painter.09} demonstrate that the environment can substantially impact on spatial structuring, for example trapping populations inside or outside regions of high anisotropy or dictating pathways of invasion. Various real world applications have been considered, including predicting the spatial spread of glioma (e.g. \cite{Painter2013}) or wolf movement along seismic lines (e.g. \cite{Hillen_Painter.13}).

The turning rate in \cite{Hillen.05}, however, was taken to be independent of orientation. To understand a consequence of this simplification, consider the migration paths of cells subjected to manufactured environments, such as surfaces subjected to micropatterning (e.g. \cite{doyle2009,thery2010}, see Figure \ref{doyleschema}) or constructed anisotropic collagen networks \cite{riching2014,Ray2017}. In \cite{doyle2009} the precise micropatterning of fibronectin on a two-dimensional surface enabled fabrication of controlled anisotropic environments, with the schematics in Figure \ref{doyleschema} (a--b) demonstrating two such arrangements. Here, cells can move from an effectively 1D region (stripes, replicating highly aligned parallel fibres) to a 2D region where the 2D regions are both isotropic, but either (a) {\em uniformly isotropic}, or (b) featuring {\em criss-crossing} perpendicular stripes. As we will explicitly show in Section \ref{s:Doyle}, the earlier transport model of \cite{Hillen.05} is unable to discriminate between these scenarios. 

If, instead, cells modulate their turning frequency by turning infrequently when moving along fibres, we can expect very distinct behaviours. Under the uniformly isotropic case, cells would be expected to meander significantly in the uniformly isotropic region, adopting short runs in any orientation. On the other hand, criss-crossing stripes could allow significant translations in either of the two dominating axial directions, hastening rediscovery of the quasi-1D regions. Indeed, evidence is found of this in the experiments of \cite{doyle2009}, {\em cf.} Supplementary Movie 7. We will explicitly show in Section \ref{s:Doyle} that direction-dependent turning can lead to enhanced effective diffusion.

\section{Model with direction-dependent turning rate}\label{model_dir_dep}

\subsection{Model formulation}

Let $p = p(t,x,w)$ denote the cell density distribution, defined at time $t\geq0$, position $x \in \Omega \subseteq \mathbb{R}^d$ and velocity $w \in V$. We typically assume $V$ to be a compact set as in \cite{Othmer_Hillen.00}, and in particular consider $V=[s_1,s_2]\times\mathbb{S}^{d-1} $ (\cite{Hillen.05}), where $\mathbb{S}^{d-1}$ is the set of all possible directions $\hat{w} \in \mathbb{R}^d$ and the hat-symbol indicates unit vectors. The limits $s_1,s_2$ denote the minimal and maximal speed ($|w|$) of the cells, with $0\leq s_1 \le s_2 <\infty$. Note that if the speed is approximately constant we simply set $V=s\mathbb{S}^{d-1}$.

The governing transport equation for describing cell movement is 
\begin{equation}\label{transport.general}
\frac{\partial p(t,x,w)}{\partial t} + w\cdot \nabla p(t,x,w) = \mathcal{L} p(t,x,w) \, ,
\end{equation}
where the operator $\nabla$ denotes the spatial gradient. The {\it turning operator} 
$\mathcal{L}$ is a linear operator that models the change in velocity of individuals per unit of time at $(x,w)$ that is not due to the free particle drift.   $\mathcal{L}$ is generally defined as an integral operator on $L^2$ spaces  \cite{Othmer_Hillen.00}:
\[
\begin{array}{lr}
\displaystyle \mathcal{L} : \displaystyle L^2(V) \longmapsto  L^2(V) \, , \\[10pt] 
\displaystyle p(t,x,w)\longmapsto \mathcal{L}p(t,x,w) \, ,
\end{array}
\]
where $(t,x)$ are independent parameters and 
\begin{equation}\label{turning.operator.mu}
\mathcal{L}p(t,x,w)=-\mu(t,x,w) p(t,x,w)+\displaystyle \int_{V} \mu(t,x,w') q(t,x,w,w') p(t,x,w') dw' \, .
\end{equation}
This operator describes the velocity scattering. As noted earlier, the key  determinants of the migration path are the {\it turning rate} function, $\mu(t,x,w)$, and the {\it turning kernel}, $q(t,x,w,w')$. Viewed in this light, the first term on the right hand side of \eqref{turning.operator.mu} models particles switching away from velocity $w$ and the second one takes into account the particles switching into velocity $w$ from all other velocities.

The turning kernel, $q(t,x,w,w')$, denotes the probability measure of switching velocity from $w'$ to $w$, given that a 
turn occurs at location $x$ and time $t$. Here we adopt the same reasoning as in \cite{Hillen.05}, by assuming reorientation is dominated by the fibrous/anisotropic environmental structure, such as collagen matrix fibres or white matter tracts. Then, the choice of new direction is derived from this structure, rather than the incoming velocity, and for $q$ we assume: 
\begin{itemize}
\item[{\bf A1}] $q(t,x,w,w')=q(t,x,w)$ depends only on the post-turning orientation;
\item[{\bf A2}] $q(t,x,w) \geq 0 \quad \forall w \in V, \quad a.e.\quad x \in \Omega, \quad \forall t \geq 0$;
\item[{\bf A3}] $q(t,x,\cdot) \in L^1(V)$ and $\displaystyle \int_{V} q(t,x,w) dw =1 \quad a.e. \quad x \in \Omega, \quad \forall t \geq 0$.
\end{itemize}
The simple assumption adopted in \cite{Hillen.05} was to directly link a probability measure describing the directional distribution of fibres, $\tilde{q}(t,x,\hat{w})$, defined on $\mathbb{R}_+\times \Omega\times \mathbb{S}^{d-1}$ and satisfying $\tilde{q}\geq 0, \displaystyle \int_{\mathbb{S}^{d-1}}\tilde{q}d\hat{w}=1$, to the turning kernel
\[
q(t,x,w)=\dfrac{\tilde{q}(t,x,\hat{w})}{\omega}, \qquad \omega=\int_V \tilde{q}(t,x,\hat{w}) dw.
\] 
Note that $q(t,x,w)$ assumes $w\in V$ while $\tilde q(t,x,\hat w)$ is defined for unit vectors only, but their only difference lies in a constant scaling factor that accounts for the difference between $V$ and $\mathbb{S}^{d-1}$. Given this, we will interchangeably call $q$ the turning kernel or the distribution of fibre orientations. As a consequence of {\bf A1}-{\bf A3}, the operator \eqref{turning.operator.mu} simplifies to
\begin{equation}\label{operator}
\mathcal{L}p(t,x,w)=-\mu(t,x,w) p(t,x,w)+q(t,x,w)\displaystyle \int_{V} \mu(t,x,w')  p(t,x,w') \, dw' \, .
\end{equation}

The turning rate function, $\mu(t,x,w)$, gives the rate at which velocity switches are made for a particle located at $x$ at time $t$, moving in direction $w$. It is at this point where we substantially diverge from \cite{Hillen.05}, 
lifting the assumptions on $\mu$ and allowing $w$-, $x$-, and $t$-dependence. Significantly, this allows the turning rate $\mu$ to depend directly on the fibre orientation $q$, for example allowing a cell to continue movement with the same velocity if it is moving in the direction of highly aligned fibres. Note that as $\mu=\mu(t,x,w)$ is a turning rate, $1/\mu(t,x,w)$ defines the mean time spent by a cell running along a linear tract with velocity $w$ between two consecutive turns performed at time $t$, location $x$. We assume:
\begin{itemize}
    \item[{\bf M1}] $\mu(t,x,\cdot) \in L^1(V), \, \forall t>0, x\in \Omega$\,;
    \item[{\bf M2}] $\dfrac{q(t,x,\cdot)}{\mu(t,x,.)}\in L^1(V), \, \forall t>0, x\in \Omega\,.$
\end{itemize}

\subsection{Statistical properties}

To analyse (\ref{transport.general}) under (\ref{operator}) we make use of a number of statistical properties of the corresponding fibre and turning distributions, such as expectations and variances. A summary of these expressions is given in Table \ref{t:distributions}.

\begin{enumerate}
\item {\bf Distribution of new directions.} We consider the distribution of newly chosen directions, $q$, with expectation 
\begin{equation}\label{mean.dir}
{\bf E}_{q}(t,x)=\displaystyle \int_{V} q(t,x,w)w \,dw. 
\end{equation}
This is also the {\it the mean new velocity after a turn} and has the variance-covariance matrix
\begin{equation}\label{orientation.tensor}
\mathbb{V}_q(t,x) = \displaystyle \int_{V} q(t,x,w)\,
(w-{\bf E}_q)\otimes(w-{\bf E}_q)\, dw\,.
\end{equation}

\item {\bf Cell mean velocity and variance.}
We introduce similar macroscopic quantities for the cell population, although we stress $p$ is not itself a probability measure. First we define the {\it macroscopic density} of the population $p$ at time $t$ and position $x$ as 
\begin{equation}\label{rho}
\bar{p}(t,x)=\displaystyle \int_{V}p(t,x,w) dw \,,
\end{equation}
and the {\it total mass} of the population in $\Omega$, 
\[
m(t)=\int_{\Omega}\bar{p}(t,x)dx\,.
\]
Note that, with no population kinetics and assuming suitably lossless boundary conditions, the total mass will be 
conserved in time. With these definitions in place we can introduce the moments of the normalized cell distribution $\hat p(t,x,w)=\dfrac{p(t,x,w)}{\bar p(t,x)}$, which will be a probability measure for all $t$ and $x$. In particular we can introduce 
\noindent the expectation
\[
{\bf E}_{\hat{p}}(t,x)=\int_V \hat{p}(t,x,w)w dw\,,
\]
which is the {\it mean velocity} of the normalized population, and
 the variance (variance-covariance matrix)
\[
\mathbb{V}_{\hat{p}}(t,x) = \displaystyle \int_{V} \hat{p}(t,x,w)\,
(w-{\bf E}_{\hat{p}})\otimes(w-{\bf E}_{\hat{p}})\, dw\,.
\]
The latter provides information on the width of the distribution $\hat{p}$ in different directions. This tensor is symmetric, but can be anisotropic, i.e. the level sets of $\hat w\mapsto \hat w^T \mathbb{V}_{\hat p} \hat w$  are ellipsoids. With this we can identify the {\it  mean velocity of the cell population} as 
\noindent 
\[
\displaystyle \int_{V}  p(t,x,w) w dw =\bar{p}(t,x){\bf E}_{\hat{p}}(t,x)
\]
and the variance-covariance of the population velocity as 
\[
\displaystyle \int_{V} p(t,x,w)\,
(w-{\bf E}_{p})\otimes(w-{\bf E}_{p})\, dw = \bar{p}(t,x)\mathbb{V}_{\hat{p}}(t,x).
\]
\item{\bf Turning part of the population.} The turning operator definition (\ref{operator}) reveals a new macroscopic quantity,  
\begin{equation}\label{pmu}
p_\mu (t,x,w) = \mu(t,x,w) p(t,x,w). 
\end{equation}
$p_\mu$ can be interpreted as the part of the cell population that moves in direction $w$ {\it and is currently turning}. Then, the total turning population per unit time is 
\begin{equation}\label{p_mu}
\bar{p}_{\mu}(t,x)=\displaystyle \int_V \mu(t,x,w) p(t,x,w) dw,
\end{equation}
which is the expression in (\ref{operator}). The turning operator \eqref{operator} can then be re-written as
\begin{equation}\label{operator.2}
\mathcal{L}p(t,x,w)=\bar{p}_{\mu}(t,x) q(x,w) -\mu(t,x,w)p(t,x,w).
\end{equation}
By normalising $p_\mu$ we can define the {\it mean incoming velocity of the turning population} as
\[
{\bf E}_{p_{\mu}}(t,x)=\dfrac{\displaystyle \int_V \mu(t,x,w) p(t,x,w) w dw}{\bar{p}_{\mu}(t,x) },
\] 
and its variance-covariance matrix accordingly.

\item{\bf Run times along the fibres.} We discover later that the stationary distributions are proportional to the ratio $q/\mu$. In fact, the corresponding distribution arises in many forthcoming calculations. Hence, we introduce 
\[ C(t,x) = \int_V\frac{q(t,x,w)}{\mu(t,x,w)} dw \]
 and the normalised distribution 
 \begin{equation}\label{Tdef}
 T(t,x,w) = \frac{1}{C(t,x)} \frac{q(t,x,w)}{\mu(t,x,w)}.
 \end{equation}
Since $\mu(t,x,w)$ is a turning rate, the function 
 \[ \tau(t,x,w)=\frac{1}{\mu(t,x,w)} \]
is the mean time spent moving in direction $w$. Then, $T(t,x,w)$ is the {\it distribution of run times along the fibres of the network} and its expectation,
\begin{equation}\label{mean.dir.fr}
{\bf E}_{T}(t,x)=\displaystyle \dfrac{1}{C(t,x)} \int_{V} \dfrac{q(t,x,w)}{\mu(t,x,w)}w \,dw \,,
\end{equation}
is the average velocity along the fibre distribution. Note that this quantity discriminates between the choice of a parabolic scaling (leading to a diffusive limit, for ${\bf E}_T\approx 0$) or a hyperbolic scaling (for ${\bf E}_T\neq 0$).
With this interpretation we can regard the following normalisation constant
\[ C(t,x) = \int_V \tau(t,x,w) q(t,x,w) dw \] 
as the mean run time between turns. We can further define the displacement vector and the mean displacement vector, respectively  
 \begin{equation}\label{chichi}
 \chi(t,x,w) = w\tau(t,x,w)\,, \qquad \bar \chi(t,x)= \int_V \chi(t,x,w) q(t,x,w) dw\,,  
 \end{equation} 
 such that 
\[ {\bf E}_{T}(t,x) = \frac{\bar \chi(t,x)}{C(t,x)} \]
becomes the ratio of the 
mean displacement vector over the mean run time on the fibre network.

\item{\bf Persistence.} {\it Persistence}, $\psi_d$, is a measure of a random walker's tendency to maintain direction during directional changes. It has values $\psi_d\in [-1,1]$, where $\psi_d=1$ denotes perfect persistence (continuing with the previous direction), $\psi_d=0$ denotes uniform turning and $\psi_d=-1$ indicates a switch into the opposite direction \cite{Alt.88, Othmer_Hillen.02}. Persistence is often computed as the {\it mean cosine} along a particle trajectory, however in our abstract framework we define it as the mean velocity of the equilibrium distribution 
$T$, i.e. 
\begin{equation}\label{adj.per.gen}
\psi_d(w)=\frac{{\bf E}_T \cdot w}{|{\bf E}_T| |w|}\,.
\end{equation}
Explicit use of the vector product shows that $\psi_d$ does indeed arise as a mean cosine, as 
\[ \psi_d = \cos(\sphericalangle({\bf E}_T, w)),\]
where $\sphericalangle({\bf E}_T, w)$ denotes the angle between ${\bf E}_T$ and $w$. 
\color{black}
We have not yet shown that $T$ is the equilibrium distribution, but do so in the next Section. Observe that if $\mu$ does not depend on $w$, Eq. \eqref{adj.per.gen} becomes
\begin{equation}\label{adj.per.gen_old}
\psi_d(w)=\frac{{\bf E}_q \cdot w}{|{\bf E}_q| |w|},
\end{equation}
as defined in \cite{Alt.88}. Hence, for turning rates that do not depend on the direction, the persistence will vanish for a bi-directional tissue ($\ie$ ${\bf E}_q=0$). Our extension to $w$-dependence in $\mu$ lifts this limitation, allowing 
non-vanishing persistence even for a bi-directional tissue with ${\bf E}_T\neq 0$. In other words, Eq. \eqref{transport.general} with \eqref{operator} is capable of generating a persistent random walk even under a 
fully symmetric configuration.
\end{enumerate}

\noindent We return to the turning operator (\ref{operator.2}). As expected, due to assumption {\bf A2}, we observe that the total cell density during turning will be conserved: 
\begin{eqnarray*}
\int_{V} \mathcal{L}p(t,x,w) dw  &=&  \int_{V} \bar p_{\mu}(t,x)  q(t,x,w)  dw  -  \int_{V} \mu(t,x,w) p(t,x,w) dw\\
&=& \bar p_\mu(t,x) - \bar p_{\mu}(t,x) \\
&=& 0.
\end{eqnarray*}
The average outgoing velocity of the total turning population, meanwhile, changes: 
\begin{eqnarray}\label{mean.vel.dyn}
\int_{V} \mathcal{L}p(t,x,w) w \, dw  
&=& \int_V\bar p_\mu(t,x) q(t,x,w) w \, dw - \int_{V} \mu(t,x,w) p(t,x,w) w \, dw  \nonumber \\
&=& \bar{p}_{\mu}(t,x){\bf E}_q(t,x) - \bar{p}_{\mu}(t,x){\bf E}_{p_{\mu}}(t,x) \nonumber\\ 
&=& \bar{p}_{\mu}(t,x)\big( {\bf E}_q(t,x)-{\bf E}_{p_{\mu}}(t,x)\big).
\end{eqnarray}
Therefore, the mean velocity of the turning population, ${\bf E}_{p_{\mu}}$, relaxes towards the average post-turning velocity, ${\bf E}_q$, imposed by the fibre network. 

\begin{table}
\footnotesize
\caption{Summary of the key probability distributions used during our analysis}.\label{t:distributions}
  \begin{minipage}{\textwidth}
    \tabcolsep=1pt
    \begin{tabular}{llll}
    \hline\hline
      distribution& meaning & expectation & variance\\
      \hline
     $\tilde q(t,x,\hat w)$ & distribution of fibre orientations &  & \\
$q(t,x,w)$ & distribution of newly chosen directions & ${\bf E}_q(t,x)$ & $ \mathbb{V}_q(t,x)$ \\
$\bar p(t,x) = \int p(t,x,w) dw$ & macroscopic cell density & & \\
$\hat p(t,x,w) = \frac{p(t,x,w)}{\bar p(t,x)}$ & normalized cell density & ${\bf E}_{\hat p}(t,x)$ & $ \mathbb{V}_{\hat p}(t,x)$\\
$ p_\mu(t,x,w)\! = \!\mu(t,x,w) p(t,x,w)  $ & turning part of the population & ${\bf E}_{p_\mu}(t,x)$ & $ \mathbb{V}_{p_\mu}(t,x)$\\
$T(t,x,w)=\frac{1}{C(t,x)}\frac{q(t,x,w)}{\mu(t,x,w)}  $ & distribution of run times along the fibres & ${\bf E}_T(t,x)$ & $ \mathbb{V}_T(t,x)$ \\
    \hline\hline
    \end{tabular}
  \end{minipage}
  \end{table}

For much of the analysis we assume fibre orientations and turning rates are time independent, i.e. $q(x,w)$ and $\mu(x,w)$. Similar arguments apply for the time dependent case, yet the notation becomes clumsier. 

\subsection{Equilibrium state of the transport equation}

As a first step towards finding equilibrium distributions, we compute the kernel of the turing operator $\mathcal{L}$.  A function $\phi(w)$ belongs to $\mbox{ker}(\mathcal{L})$ if and only if it satisfies
\[
\begin{array}{lr}
\mathcal{L}\phi(w)=0\\[8pt]
\Longleftrightarrow -\mu(x,w) \phi(w)+q(x,w)\displaystyle \int_{V} \mu(x,w')  \phi(w') dw'=0\\[8pt]
\Longleftrightarrow \mu(x,w) \phi(w)= q(x,w)\displaystyle \int_{V} \mu(x,w')  \phi(w') dw'.
\end{array}
\]
We can write
\[
\phi(w)=\frac{q(x,w)}{\mu(x,w')}  \bar{\phi}_{\mu}(x)\,,
\]
where $\bar{\phi}_{\mu}$ is a $w$-independent function 
\[
\bar{\phi}_{\mu}(x)=\displaystyle \int_{V} \mu(x,w) \phi(w) dw, 
\]
which is the fraction of turning cells of the population $\bar{\phi}$. The $w$-dependence of elements of $\mbox{ker}(\mathcal{L})$ is given by the ratio $q(x,w)/\mu(x,w)$, i.e. it is given by the run-time distribution along the fibres, $T(x,w)$, which we introduced in (\ref{Tdef}). Then 
\[ \phi(x,w) = \bar\phi(x) T(x,w), \qquad \mbox{and} \qquad \mbox{ker}(\mathcal{L}) = \langle T (x,w) \rangle. \]
Hence a stationary state of the equation \eqref{transport.general}-\eqref{operator} has the form
\begin{equation}\label{Maxwellian}
M(x,w)= \beta(x) T(x,w) 
\end{equation}
which we call the Maxwellian.   
We remark that \eqref{transport.general} with \eqref{operator} is the linear Boltzmann equation. In particular, the 
quantity  
\[
\sigma(t,x,w,w')=q(t,x,w)\mu(t,x,w')
\]
is the \textit{cross section} and the equilibrium probability density (normalized to 1) given by $\mathcal{L}(T)=0$ is \eqref{Tdef}. The function $T$ will be nonnegative, since $q$ and $\mu$ are nonnegative, and it is an $L^1(V)$ function via assumption ${\bf M2}$. Therefore, the stochastic process ruled by $\sigma$ satisfies micro-reversibility, i.e.
\[
\sigma(x,w,w')T(x,w')=\sigma(x,w',w)T(x,w).
\]
Consequently, existence and uniqueness of a nonnegative solution $p\in L^1(\Omega\times V)$ of \eqref{transport.general},\eqref{operator} with initial condition $p^0 \in L^1(\Omega \times V)$ and  non-absorbing boundary conditions \cite{Cercignani} is a classical result of kinetic theory (see, for example, \cite{Petterson1983}). We in particular consider a special class of non-absorbing boundary conditions, given by no-flux boundary conditions \cite{Plaza}.

Arguing as in \cite{Petterson,Bisi.Carrillo.Lods}, we can prove a linear version of the classical H-Theorem for the linear Boltzmann equation \eqref{transport.general}, \eqref{operator} with initial condition $p^0=p(0,x,w) \in \, L^1(\Omega \times V)$. Let us introduce the {\it entropy} $S=-H$ where, for any given convex function $\Phi:\mathbb{R}_+\rightarrow \mathbb{R}_+$,
\[
\ H_{\Phi}[p|M](t)= \int_\Omega \int_V  p(t,x,w) \Phi\left( \dfrac{p(t,x,w)}{M(x,w)}\right) \, dw \, dx , \qquad p(t,\cdot, \cdot) \in L^1(\Omega \times V)\,.
\]
In the above, $M$ is the Maxwellian satisfying
\[
\displaystyle \int_{\Omega \times V} M(x,w) \, dx dw = \displaystyle \int_{\Omega} \bar{p}\,^0(x) dx
\]
where we assume that $p^0$ has finite mass $\displaystyle \int_{\Omega} \bar{p}\,^0 \, dx$  and entropy $H_{\Phi}[p^0|M](0)$. Hence, $\beta(x)$ is such that $\displaystyle \int_{\Omega} \beta(x) dx =\displaystyle \int_{\Omega} \bar{p}^0(x) \, dx $ and, in particular, if a stationary state $p^{\infty}(x,w)$ exists, then $\beta(x)=\bar{p}^{\infty}(x)$, so that
\[
p^{\infty}(x,w)=\bar{p}^{\infty}(x) T(x,w).
\]
Under non-absorbing boundary conditions, via exactly the same procedure followed in \cite{Petterson} it is possible to prove that
\[
\frac{dH_{\Phi}[p|M]}{dt}(t) \leq 0 \quad \mbox{for} \quad t\ge 0 \qquad
\]
and
\[
\frac{dH_{\Phi}[p|M]}{dt} (t) =0, \,  \qquad \mbox{iff} \qquad p(t,x,w)=M(x,w),
\]
where $p(t,x,w)$ is the unique $L^1$ solution to \eqref{transport.general},\eqref{operator} with initial condition $p^0$ and non-absorbing boundary conditions. Furthermore, continuing to argue as in \cite{Petterson}, it can be proved that provided $\displaystyle \int_{\Omega}\int_{V}\big( 1+w^2 +|\log p_0|\big) \, dw dx < \infty$, then
\[
\lim_{t\rightarrow \infty} \displaystyle \int_{\Omega} \int_V |p(t,x,w)-M(x,w)| \, dw \, dx =0.
\]

\section{Macroscopic limits}\label{mod-turn_macro}

In this section we consider two distinct scalings for the transport equation: $i)$ the {\em parabolic scaling}, and $ii)$ the {\em hyperbolic scaling}. The former applies in a diffusion-dominated case while the latter corresponds to the drift-dominated case. The cases differ through the relative scaling of time and space in a suitably small parameter $\ep$. 

In the parabolic scaling we consider a small parameter, $\ep \ll 1$, and assume macroscopic time and space scales $(\tau, X)$ that scale according to 
\[
\tau=\ep^2 t,\qquad  X=\ep x .
\]
A paradigm for scaling in this manner can be found by comparing the microscopic scales of run and tumble movements 
of {\em E. coli} bacteria and the experimental scales at which population-level phenomena form, such as
travelling bands \cite{adler} or cellular aggregates \cite{budrene1991}. Runs are characterised with a movement speed typically around 10-20 $\mu m/s$ with a tumble taking place every second or so. Large scale patterning phenomena typically arise after a few hours, i.e. $O(10^4)$ seconds. Hence, $\ep^2$ = mean run time/experimental time $= 10^{-4}$ and, in turn, $\ep = 10^{-2}$. Given the micron spatial scale of individual movement, the corresponding macroscopic scale is $10^2 \times 10 \mu m = 1 mm$, which is the scale of the cell aggregates that typically form.   

This scaling therefore demands a sufficiently slow time scale over which diffusion can begin to dominate or, 
equivalently, when cells have large speeds and turning rates. Rescaling (\ref{transport.general}),(\ref{operator}) 
gives
\begin{equation}\label{eq.par}
\begin{array}{lr}
\ep^2\dfrac{\partial p}{\partial \tau}(\tau,X,w) + \ep w\cdot \nabla p(\tau,X,w) =\\[8pt]
 -\mu(X,w) p(\tau,X,w)+q(X,w)\displaystyle \int_{V} \mu(X,w')  p(\tau,X,w') dw' \, ,
\end{array}
\end{equation}
where the gradient $\nabla$ is now applied with respect to $X$. We assume that the rescaled turning frequency and kernel are such that $\mu(X,w), q(X,w) \sim \mathcal{O}(1)$. Thanks to classical results (see e.g. \cite{Petterson1983}), we have existence and uniqueness of the solution of \eqref{eq.par} with initial condition $p^0 \in L^1(\Omega \times V)$ and non-absorbing boundary conditions. We now consider an asymptotic expansion of $p$ in orders of $\ep$, 
\begin{equation}\label{expansion}
p(\tau,X,w)=p_0(\tau,X,w)+\ep p_1(\tau,X,w)+\ep^2 p_2(\tau,X,w)+\mathcal{O}(\ep^3),
\end{equation}
and we are particularly interested in the leading order term $p_0$. 

\subsection{Diffusion-Dominated Case}

In the diffusion-dominated case we assume that the macroscopic drift term ${\bf E}_T=0$. We consider this case first to introduce the necessary technical notations. 

\begin{theorem}\label{t:paralim}
Let assumptions {\bf A1}-{\bf A3} be satisfied and assume 
\begin{equation}\label{nec.diff.cond.}
{\bf E}_T=0.
\end{equation}
Consider equation (\ref{eq.par}) with expansion (\ref{expansion}). The leading order term, $p_0(\tau,X,w)$, of equation (\ref{expansion}) satisfies
\[
p_0(\tau,X,w)=\bar p_0(\tau,X) \, T(X,w),\]
where $\bar p_0(\tau,X)$ solves the macroscopic anisotropic diffusion equation
\begin{equation}\label{paralim}
\frac{\partial}{\partial \tau} \bar p_0(\tau,X) = \nabla \cdot \int_V \frac{1}{\mu(X,w)}\nabla\cdot\Bigl[\bar p_0(\tau,X) 
 \mathbb{D}_{T}(X,w)\Bigr] dw\,,
 \end{equation}
with microscopic anisotropic diffusion tensor 
\[
\mathbb{D}_{T} (X,w)= w \otimes w\, T(X,w).
\]
\end{theorem}

\begin{proof}

We have seen that 
\[ \mbox{ker}(\mathcal{L}) = \langle T\rangle. \]
In order to invert $\mathcal{L}$ on the orthogonal complement of its kernel, we use a specific weight for the inner product in $L^2_\eta(V)$, with 
\begin{equation}\label{defeta}
\eta(X,w) = \frac{\mu(X,w)}{T(X,w)},
\end{equation}
$\ie$ given $f,g \in L^2$ the weighted scalar product is defined by
\[
(f,g)_{\eta}=\int_V f(w)g(w)\eta(X,w) \, dw.
\]
We observe that $\phi(X,w)\in \langle T \rangle^\perp$ if and only if 
\[ 0= \int_V \phi(X,w) T(X,w) \frac{\mu(X,w)}{T(X,w)}dw = \int_V \phi(X,w) \mu(X,w) dw = \bar\phi_\mu(X). \]
The range of $\mathcal{L}$ is given by all functions that integrate to zero, since 
\[ \int_V \mathcal{L} \phi(X,w) dw =\overline{\mathcal{L}\phi} =0\]
from mass conservation. Then, this range, as subset of $L^2_\eta(V)$, can be written as 
\[ \mbox{Range}(\mathcal{L}) = \langle \eta^{-1}\rangle^\perp,\]
since 
\[ 0= \int_V \phi(X,w) dw = \int_V\phi(X,w) \eta^{-1}(X,w) (\eta(X,w) dw) = (\phi, \eta^{-1})_\eta.\] 
Then 
\[ \mathcal{L}^\perp: \langle T\rangle^\perp \to \langle \eta^{-1}\rangle^\perp, \qquad \mbox{with} \quad \mathcal{L}^\perp =\mathcal{L}|_{\langle T\rangle^\perp}\]
and this restricted operator $\mathcal{L}^\perp$ is the operator we would like to invert. 
 Given $\psi\in\langle \eta^{-1}\rangle^\perp$, we need to find $\phi \in \langle T\rangle^\perp$ such that  $\mathcal{L}\phi = \psi$. Applying $\mathcal{L}$ from (\ref{operator.2}) we obtain (ignoring arguments for clarity of presentation) 
\[
\mathcal{L} \phi = \bar \phi_\mu q - \mu \phi = \psi.
\]
Since on $\langle T\rangle^\perp$ we have $\bar \phi_\mu=0$, we can solve for $\phi$ as 
\[
\phi(X,w) = -\frac{1}{\mu(X,w)} \psi(X,w) .
\]
Hence we can write 
\begin{equation}\label{Linverse}
\left(\mathcal{L}^\perp \right)^{-1}:\langle\eta^{-1}\rangle^\perp \to \langle T \rangle^\perp, \quad \psi\mapsto   -\frac{1}{\mu} \psi.
\end{equation}

\noindent We may observe that the pseudo-inverse depends on the microscopic velocity and on the macroscopic space coordinate through $\mu(X,w)$. \\

\noindent We now substitute expansion (\ref{expansion}) into equation (\ref{eq.par}) and match orders of $\ep$.

\begin{itemize}
    \item  For  $\ep^{0}$ we find 
\[
\mathcal{L}p_{0}(\tau,X,w)=0,
\label{eps0q.mu}
\]
hence $p_0\in \mbox{ker}(\mathcal{L})$ and we can write 
\begin{equation}\label{leadingorder}
p_0(\tau, X, w) = \bar{p}_0(\tau, X)\; T(X, w) .
\end{equation}

\item At order $\ep^{1}$ we have  
\begin{equation}\label{firstorder}
\nabla\cdot\big( w\, p_{0}(\tau,X, w)\big)= \mathcal{L} p_1(\tau,X,w)\,.
\end{equation}
To solve this equation for the next order correction term $p_1$, we need to invert $\mathcal{L}$. We saw earlier in (\ref{Linverse}) that $\mathcal{L}$ is invertible on $\langle \eta^{-1}\rangle^\perp $. Hence we check the solvability condition
\[ \nabla\cdot(wp_0(\tau,X,w)) \in \langle\eta^{-1} \rangle^\perp.\]
Using (\ref{leadingorder}) this condition becomes
\begin{eqnarray*}
0 &=& (\nabla\cdot(wp_0) , \eta^{-1})_\eta \\
&=& \int_V \nabla\cdot(w\bar p_0(\tau,X) T(X, w)) \eta(X,w)^{-1} \eta(X,w)  dw \\
&=& \nabla \cdot \left[ \int_V w T(X,w) dw \,\bar p_0(\tau, X)\right] \\
&=& \nabla\cdot \left[ \mathbb{E}_T(X) \bar p_0(\tau,X)\right]\,.
\end{eqnarray*}
Hence in this step it is necessary to assume
\begin{equation}\label{ETzero} 
\mathbb{E}_T(X) = 0.
\end{equation}
In words, we assume the distribution of run times along the fibre network has no dominant direction.  In this case we can invert (\ref{firstorder}) and find 
\begin{equation}\label{p1}
p_1(\tau,X, w) = \frac{1}{\mu(X, w)} \nabla\cdot(w\, \bar p_0(\tau,X) \, T(X,w))\,.
\end{equation}

\item In $\ep^{2}$:
\[
\dfrac{\partial }{\partial \tau}p_{0}(\tau,X, w)+\nabla\cdot\big(p_{1}(\tau,X, w)w \big)=\mathcal{L} p_2(\tau,X,w)\,.
\] 
Integrating this equation over $V$ we obtain 
\[ \int_V \frac{\partial }{\partial\tau}p_0(\tau, X, w) dw +\nabla\cdot \int_V w p_1(\tau,X, w) dw =0.\]
Using the expressions (\ref{leadingorder}) for $p_0$ and (\ref{p1}) for $p_1$, we obtain 
\begin{eqnarray*}
\dfrac{\partial }{\partial \tau}\bar p_0(\tau,X) \underbrace{\int_V T(X,w) dw}_{=1} &=& \nabla \cdot \int_V \frac{1}{\mu(X,w)} w\otimes w  \nabla \Bigl[\bar p_0(\tau, X) T(X, w) \Bigr] dw \\
 &=& \nabla \cdot\int_V \frac{1}{\mu(X,w)}\nabla\cdot\Bigl[\bar p_0(\tau,X) w\otimes w T(X, w)\Bigr] dw \\
 &=& \nabla \cdot \int_V \frac{1}{\mu(X,w)}\nabla\cdot\Bigl[\bar p_0(\tau,X) 
 \mathbb{D}_{T}(X,w)\Bigr] dw\,,
 \end{eqnarray*}
with an anisotropic diffusion tensor 
\[
\mathbb{D}_{T} (X,w)= w \otimes w\, T(X,w)\,.
\]
$\mathbb{D}_T$ is a macroscopic scale diffusion tensor, where 
\[\mathbb{V}_T(X) = \int_V \mathbb{D}_T(X,w) dw \]
describes the variance of the run time distribution along the network fibres. 
\end{itemize}

\end{proof}

\begin{lemma}\label{l:paralim}
The above parabolic limit equation (\ref{paralim}) can be written as an anisotropic drift-diffusion model 
\begin{equation}\label{FAAD.xv}
\dfrac{\partial \bar{p}_0}{\partial \tau}(\tau,X) +\nabla \cdot \big(\mathbf{a}(X)\bar{p}(\tau,X)\big)= \nabla \cdot \nabla \cdot \big(\mathbb{D}(X)\bar{p_0}(\tau,X)\big) ,
\end{equation}
with macroscopic diffusion tensor, $\mathbb{D}$, and advection speed, $\mathbf{a}$, given by 
\begin{eqnarray}
\mathbb{D}(X) &=& \int_V w\otimes w\, \frac{T(X,w)}{\mu(X,w)} dw\label{macroDp} \\
\mathbf{a}(X) &=& - \int_V w\otimes w\, \frac{\nabla\mu(X,w)}{\mu(X,w)}\frac{T(X,w)}{\mu(X,w)} dw.\label{macroap}
\end{eqnarray}
\end{lemma}
\begin{proof} The proof relies on straightforward application of the quotient rule. Omitting arguments for readability\,,  
\[ \nabla\cdot\nabla\cdot \left(\bar p_0 \int_V w\otimes w \frac{T}{\mu} dw\right) = \nabla \cdot \int_V \frac{1}{\mu} \nabla\cdot (\bar p_0 w\otimes w T) dw  - \nabla\cdot \left(\bar p_0 \int_V \frac{\nabla \mu}{\mu^2} w\otimes w T dw \right). \]
\end{proof}

\noindent This representation allows some interesting physicobiological interpretations.
\begin{itemize}
    \item In (\ref{defeta}) earlier we defined the weight function $\eta=\frac{\mu}{T}$. We can write the above macroscopic quantities in terms of $\eta$ as 
  \begin{eqnarray*}
\mathbb{D}(X) &=& \int_V w\otimes w\, \frac{dw}{\eta(X,w)}\,, \\
\mathbf{a}(X) &=& - \int_V w\otimes w\, \nabla\ln(\mu(X,w))\frac{dw}{\eta(X,w)}\,.
\end{eqnarray*}  
$\mathbb{D}$ then appears as an anisotropy matrix related to the measure $\eta^{-1}dw $, while $\mathbf{a}$ measures the anisotropic logarithmic gradient with the same measure $\eta^{-1} dw$.
\item We can also relate the terms back to the original network structure given by $q(X,w)$. Recall that $T=\frac{q}{C\mu}$, where $C(X)$ was a normalisation constant.  By considering the distance travelled in direction $w$ we can write 
\begin{eqnarray*}
\mathbb{D}(X) &=&\frac{1}{C(X)} \int_V \chi(X,w) \otimes \chi(X,w) q(X,w)\, dw\,, \\
\mathbf{a}(X) &=& - \frac{1}{C(X)} \int_V \chi(X,w) \otimes \chi(X,w)\, \nabla\ln(\mu(X,w)) \, q(X,w)  dw\,,
\end{eqnarray*}
using the definition of $\chi$ from (\ref{chichi}). 
$\mathbb{D}$ is then the scaled variance-covariance matrix of the directed mean run-time along the fibre network and $\mathbf{a}$ is the scaled mean logarithmic derivative of the turning rate, weighted by the mean run times along the fibre network. 
\item The equations simplify drastically when the turning rate $\mu$ does not depend on space. In this case $\nabla\mu =0$ and hence $\mathbf{a}=0$, i.e., we get a pure anisotropic diffusion equation. Viewed this way, the advection velocity $\mathbf{a}$ clearly results from spatial variation in the turning rate, $\mu(X,w)$, and relates to an anisotropic taxis term measuring the drift due to the gradient of $\mu$. Spatial dependence in $\mu(X,w)$ generates an advection pushing cells towards decreasing values of the turning frequency. 
\end{itemize}

Due to the fact that $V$ and $\Omega$ are compact, we can directly apply a convergence result of \cite{DegondGoudonPoupaud}, giving us
\begin{lemma}
Suppose {\bf A1}-{\bf A3} and {\bf M1}-{\bf M2} hold and assume \eqref{nec.diff.cond.}. Let us also suppose that $\exists \, C_1,C_2 \ge0$ such that
\[
|w\cdot\nabla q(x,w)| \le C_1\mu(x,w), \quad a.e. \quad \textit{in} \quad \Omega\times V
\]
and 
\[
\displaystyle\int_V w^2 \dfrac{q(x,w)}{\mu(x,w)^2} dw \le C_2 \quad a.e. \quad in \quad \Omega.
\]
Let the initial condition $p^{0}_{\epsilon}(x,w)$ satisfy $\displaystyle \int_{\Omega \times V} \dfrac{p_{{\epsilon}^0}(x,w)^2}{T(x,w)}\, dx dw < \infty$ and $\bar{p}_{\epsilon}^0 \rightharpoonup \bar{p}^0$ in $H^{-1}(\Omega)$. Let  $p_{\epsilon}$ be the solution to \eqref{eq.par}. Then there exists a subsequence  $\bar{p}_{\epsilon} \rightharpoonup \bar{p}$ in $L^2((0,T)\times \Omega)$ where $\bar{p}$ satisfies \eqref{FAAD.xv} with initial condition  $\bar{p}_0(x)$.
\end{lemma}

\subsection{Drift-Diffusion Case}

The authors of \cite{GoudonMellet} study the parabolic scaling also in the case where the macroscopic drift ${\bf E}_T \ne 0$. The key lies in a transformation to moving spatial coordinates, shifting the solution in the direction of ${\bf E}_T$ as $Z=X-{\bf E}_T t$. This method also works here and we introduce 
\[ p(\tau,X,w) := u(\tau, X-{\bf E}_T \tau, w). \]
Then, the transport equation \eqref{transport.general} transforms to
\[ \frac{\partial}{\partial \tau}u +\nabla\cdot [ (w-{\bf E}_T) u] = \mathcal{L} u.\]
Therefore, the scaled transport equation (\ref{eq.par})  
\[\ep^2 \frac{d}{d\tau} p +\ep \nabla \cdot (w p) = \mathcal{L} p \]
transforms to 
\[ \ep^2 \frac{\partial}{\partial \tau}u +\ep \nabla\cdot [ (w-{\bf E}_T) u] = \mathcal{L} u.\]
For this modified transport equation we perform the same scaling analysis as before. We consider 
\[ u(\tau,Z,w) = u_0(\tau, Z, w) +\ep u_1(\tau,Z, w) +\ep^2 u_2(\tau, Z, w) + \cdots. \]
Upon comparing orders of $\ep$ we find, to leading order, that 
\[ \mathcal{L} u_0 = 0.\]
Hence, $u_0$ is in the kernel of $\mathcal{L}$ and we can write 
\[ u_0(\tau,Z, w) = \bar u_0(\tau,Z) \; T(X, w). \]
The order $\ep$ terms are 
\[ \nabla\cdot\big[(w-{\bf E}_T) \bar u_0 T\big] = \mathcal{L} u_1.\]
To solve this equation for $u_1$ we require the solvability condition 
\[
\nabla\cdot \left[\int_V (w-{\bf E}_T) T dw \; \bar u_0 \right] =0\,,\]
which is true since 
\[ \int_V(w-{\bf E}_T) T dw = {\bf E}_T - {\bf E}_T =0.\]
Then 
\[ u_1 = -\frac{1}{\mu} \nabla\cdot\big[(w-{\bf E}_T) \bar u_0 T\big].\]
The terms of order $\ep^2$ are 
\[\frac{\partial}{\partial \tau} (\bar u_0 T) + \nabla\cdot \big[(w-{\bf E}_T) u_1\big] = \mathcal{L} u_2.\]
We integrate this equation over $V$ to obtain 
\begin{equation}\label{intermediate1}
0= \frac{\partial}{\partial \tau}\bar u_{0} - \nabla\cdot \underbrace{\int_V (w-{\bf E}_T) \left(\frac{1}{\mu}\right) \nabla\cdot \big[(w-{\bf E}_T) \bar u_0 T\big] dw}_{(I)}. 
\end{equation} 
Generalising the previous definitions of the diffusion tensor (\ref{macroDp}) and the drift velocity (\ref{macroap}) we define now a macroscopic diffusion tensor $\mathbb{D}_h$ and macroscopic drift velocity $\mathbf{a}_h$ as 
\begin{eqnarray}
\mathbb{D}_h(X) &=& \int_V (w-{\bf E}_T)\otimes (w-{\bf E}_T)\, \frac{T}{\mu} dw\,,\label{macroDh} \\
\mathbf{a}_h(X) &=&  \int_V (w-{\bf E}_T) \nabla \cdot\left[\frac{w-{\bf E}_T}{\mu}\right] T\; dw\,.\label{macroah}
\end{eqnarray}
Note that for ${\bf E}_T=0$ we return to the previous definitions (\ref{macroDp}) and (\ref{macroap}) of the diffusion-dominated case. 

Therefore we have that
\begin{eqnarray}
\nabla \cdot (\mathbb{D}_h \bar u_0) &=& \int_V (-\nabla \cdot {\bf E}_T) (w-{\bf E}_T) \frac{T}{\mu} dw\; \bar u_0  +\underbrace{\int_V(w-{\bf E}_T) \frac{1}{\mu} \nabla \cdot\big[(w-{\bf E}_T) \bar u_0 T\big] \, dw}_{(I)}\nonumber \\
&&+ \int_V (w-{\bf E}_T) \otimes (w-{\bf E}_T) T \nabla\left(\frac{1}{\mu}\right) dw\; \bar u_0\nonumber\\
&=& (I) + \int_V(w-{\bf E}_T) \left[\nabla\cdot(w-{\bf E}_T) \frac{1}{\mu} + (w-{\bf E}_T)  \cdot\nabla\left(\frac{1}{\mu}\right) \right] Tdw \;\bar u_0 \nonumber \\
&=& (I) + {\bf a}_h \bar u_0\,. \label{IDa}
\end{eqnarray}
With these definitions we obtain from (\ref{intermediate1}) a fully anisotropic drift-diffusion model for $\bar u_0$:
\begin{equation}
    \frac{\partial}{\partial \tau}\bar u_{0} +\nabla\cdot ({\bf a}_h \bar u_0) = \nabla\cdot\nabla\cdot (\mathbb{D}_h \bar u_0) \,.
\end{equation}
Finally, transforming back to $\bar p_0 (\tau, X) = \bar u_0(\tau, X-{\bf E}_T \tau)$ we find 
\begin{equation}\label{eq:par_drift}
    \frac{\partial}{\partial \tau}\bar p_{0} +\nabla\cdot (({\bf a}_h+{\bf E}_T) \bar p_0) = \nabla\cdot\nabla\cdot (\mathbb{D}_h \bar p_0) \,.
\end{equation}
If ${\bf E}_T=0$ we get the same result as already shown in Theorem \ref{t:paralim} and equation (\ref{FAAD.xv}) in Lemma \ref{l:paralim}.

\subsection{Hyperbolic limit}

In drift-dominated phenomena we expect that nondimensionalisation leads to a scaling in which the macroscopic time and space scales are of the form $\tau=\ep t, X=\ep x$, $\ep \ll 1$. Effectively, cells do not have a large turning frequency, drift dominates and we can perform a hyperbolic limit. The rescaled equation is
\begin{equation}\label{drift.resc}
\ep \dfrac{\partial p}{\partial \tau} + \ep w\cdot \nabla p =\mathcal{L} p \,. 
\end{equation}
We consider the following expansion
\begin{equation}\label{hyp.p}
p= p_0 + \ep g +\mathcal{O}(\ep^2),
\end{equation}
where we assume that the correction term $g$ carries zero mass ($\bar g =0$) and, therefore, $\bar p_0=\bar p$.  
Substituting (\ref{hyp.p}) into ($\ref{drift.resc}$) we find that the leading order terms ($\ep^{0}$) are $\mathcal{L} p_0= 0$, hence 
\begin{equation}\label{hypone}
p_0(\tau,X, w) = \bar p(\tau,X)\, T(X,w)\,.
\end{equation}
With this choice of $p_0$ the remaining terms in (\ref{drift.resc}) are
\begin{equation}\label{hypfour}
\dfrac{\partial p}{\partial \tau} + \ep g_\tau + w\cdot \nabla p + \ep w\cdot \nabla g = \mathcal{L} g + O(\ep^2).
\end{equation}
We integrate this equation over $V$, use the form of $p_0$ from above (\ref{hypone}), and assume the $O(\ep^2)$-terms are negligible. Then 
\[ \dfrac{\partial \bar p}{\partial \tau} +\ep\int_V g_\tau dw + \nabla\cdot \int_V w \bar p T dw + \ep \nabla\cdot \int_V w gdw = 0 .\]
Since $g$ carries no mass, we have $\frac{\partial}{\partial \tau} \bar g =0$. Using the definition of the expectation of $T$ from (\ref{mean.dir.fr}) we obtain \begin{equation}\label{macro.drift.1}
 \dfrac{\partial \bar p}{\partial \tau}+\nabla \cdot ( {\bf E}_T \bar{p}) + \ep \nabla\cdot \int_{V} w g dw =0 \,,
\end{equation}
which, to leading order, becomes the pure drift model 
\begin{equation}\label{hypthree}
\dfrac{\partial \bar p}{\partial \tau}(\tau, X)  +\nabla\cdot \Bigl({\bf E}_T(X) \bar p(\tau,X)\Bigr) =0 .
\end{equation}

The macroscopic drift velocity ${\bf E}_T$ is the expected movement direction based on the average time spent on the fibres. If ${\bf E}_T$ is of order one then (\ref{hypthree}) is the leading order model. However, for small or zero expectation ${\bf E}_T$, we can compute the next-order correction term $g$. Specifically, we use  $\displaystyle \int_V wg \, dw= \displaystyle \int_V (w-{\bf E}_T)g \, dw$ and rewrite equation \eqref{macro.drift.1} as
\begin{equation}\label{nextordercorrection}
\dfrac{\partial \bar{p}}{\partial \tau}+\nabla \cdot \big( {\bf E}_T\bar p \big)\\
= -\ep \nabla \cdot\Bigl[\int_V(w-\mathbf{E}_T) g dw\Bigr]. 
\end{equation}
From the previous equation (\ref{hypfour}) we find to leading order that 
\begin{equation}\label{Lequation}
\mathcal{L} g = \dfrac{\partial p}{\partial \tau} +w\cdot \nabla p 
 = \dfrac{\partial \bar{p}}{\partial \tau} T + w\cdot \nabla (\bar p T) 
 = -\nabla\cdot({\bf E}_T\bar p) T + w\cdot \nabla(\bar p T), 
\end{equation}
where we used the drift model (\ref{hypthree}) in the last step. To solve for $g$ we need to satisfy the solvability condition 
\begin{eqnarray}\label{intcond} 
0 &=&  \int_V -\nabla \cdot (\mathbf{E}_T \bar p) T + w\cdot \nabla (\bar p T) dw \nonumber\\
&=& -\nabla\cdot({\bf E}_T \bar p) +\nabla \cdot\left(\int_V w T dw \bar p\right)\nonumber \\
&=&  -\nabla\cdot({\bf E}_T \bar p)  +\nabla\cdot({\bf E}_T \bar p) \,,
\end{eqnarray}
which is true since $\int T dw = 1$. 
Hence the right hand side of equation (\ref{Lequation}) is in the range of $\mathcal{L}$. To solve an equation of the form $\mathcal{L} g = R$ with $R\in \mbox{Range}(\mathcal{L})$, we can use the pseudoinverse of $\mathcal{L}^\perp$ and add a term which lies in the kernel of $\mathcal{L}$, i.e. we write 
\[ g(\tau,X,w) = -\frac{1}{\mu(X,w)} R(\tau,X,w) + \alpha(\tau,X) T(X,w), \]
where $\alpha(\tau,X)$ is independent of $w$. Here we assumed the correction term has no mass, $\bar g =0$, and hence we choose $\alpha(\tau, X)$ in such a way that $\bar g =0$. 
We solve (\ref{Lequation}) for the correction term, $g$, to obtain
\[
g = -\frac{1}{\mu} \Bigl[-\nabla \cdot({\bf E}_T \bar p) T + w \cdot \nabla (\bar p T)   \Bigr] + \alpha T 
\]
with 
\begin{equation}\label{alpha_hyp}
\alpha(\tau,X) = -\int_V \frac{1}{\mu} \Bigl[-\nabla \cdot({\bf E}_T \bar p) T + w \cdot \nabla (\bar p T)   \Bigr]dw.
\end{equation}
Note that if the turning rate is independent of velocity $w$ (i.e.  $\mu(X,w)=\mu(X)$), then from (\ref{intcond}) we note $\alpha=0$ and return to the standard case discussed in the introduction (see \cite{Hillen_Painter.08}). We write $g$ in the following form, more convenient for later manipulation:
\begin{eqnarray*}
&g=& -\frac{1}{\mu}\Big[ T(w-{\bf E}_T ) \cdot \nabla \bar p + \bar p(w\cdot \nabla T) - \bar p_0 ((\nabla\cdot {\bf E}_T) T ) \Big] + \alpha T \\
&=& -\frac{1}{\mu}\Big[ T(w-{\bf E}_T) \cdot \nabla \bar p + \nabla T \cdot (w-{\bf E}_T) \bar p  + {\bf E}_T \cdot \nabla T\bar p  + \big[\nabla\cdot(w-{\bf E}_T) \big]\bar p T\Big] + \alpha T \\
&=& -\frac{1}{\mu} \Big[\nabla\cdot\big[(w-{\bf E}_T) \bar p T \Big] + {\bf E}_T\cdot \nabla T \bar p \Big] + \alpha T.
\end{eqnarray*}
Then, the term in the square brackets of (\ref{nextordercorrection}) becomes 
\begin{eqnarray}
 \int_V(w-\mathbf{E}_T) g dw &=& -\underbrace{\int_{V}(w-{\bf E}_T)\frac{1}{\mu}\nabla\big[(w-{\bf E}_T) \bar p T\big] dw}_{(I)}\\
 && 
 - \int_V \frac{1}{\mu} (w-{\bf E}_T) {\bf E}_T \cdot \nabla T \, dw\; \bar p + \int_V(w-{\bf E}_T) \alpha T \, dw \\
&=& -\nabla \cdot(\mathbb{D}_h \bar p ) + {\bf a}_h \bar p - \int_V  (w-{\bf E}_T) {\bf E}_T \cdot \frac{\nabla T}{\mu} \, dw\; \bar p\,.
\end{eqnarray}
In the above we used the relation (\ref{IDa}) between the integral $(I)$ and the diffusion and drift terms $\mathbb{D}_h$ and ${\bf a}_h$, respectively, as well as the fact that 
\[ \int_V(w-{\bf E}_T) \alpha T dw = \alpha \left(\int_V w T dw - {\bf E}_T \int_V T dw \right) =0. \]
Combining these calculations with the macroscopic limit equation (\ref{nextordercorrection}), we obtain the hyperbolic limit equation with correction term as
\begin{equation}\label{macro.drift.3}
\dfrac{\partial \bar{p}}{\partial \tau}+\nabla \cdot \big( ({\bf E}_T +\ep {\bf a}_h)  \bar{p}\big)=\ep \nabla \cdot \nabla \cdot \Big( \mathbb{D}_h\bar{p}\Big)+\ep\nabla \cdot
\left(\int_V  (w-{\bf E}_T) {\bf E}_T \cdot \frac{\nabla T}{\mu} \, dw\; \bar p\right).
\end{equation}
where $\mathbb{D}_h$ and $\mathbf{a}_h$ are given by (\ref{macroDh}) and (\ref{macroah}), respectively. 

A particularly pleasing aspect of the above hyperbolic limit lies in its generalisation of the earlier parabolic limit. 
Specifically, for the case ${\bf E}_T=0$ we obtain the same macroscopic quantities: (\ref{macroDh}) coincides with (\ref{macroDp}) and (\ref{macroah}) coincides with (\ref{macroap}). Moreover, for ${\bf E}_T=0$, the rather clumsy 
integral correction term vanishes and we obtain a rescaled version of the parabolic limit equation (\ref{FAAD.xv}): 
\begin{equation}
\dfrac{\partial \bar{p}}{\partial \tau}+\ep \nabla \cdot \big( ( {\bf a}_h)  \bar{p}\big)=\ep \nabla \cdot \nabla \cdot \Big( \mathbb{D}_h\bar{p}\Big)\,.
\end{equation}
A further scaling of time with $\ep$ then reproduces the parabolic limit (\ref{FAAD.xv}).

\subsection{Time-dependent turning distribution}

For time-varying tissues, $q$, $\mu$ and $C$ will all depend on time. Rescaling \eqref{transport.general},\eqref{operator} with $X=\ep x$, regardless of using time scale $\tau=\ep^2 t $ or $\tau= \ep t$,  and comparing equal orders of $\ep$ allows us to obtain the leading order function of \eqref{expansion}, that is 
\[
p_0(\tau,X,w)=\bar{p}(\tau,X)T(\tau,X,w)
\]
where the equilibrium distribution is now also time dependent, $T(\tau,X,w)=\dfrac{q(\tau,X,w)}{\mu(\tau,X,w)}$.

The diffusive limit does not change, but the hyperbolic limit has different correction terms. Proceeding as in 
the previous section, we find that the correction $g$ is a solution to
\begin{equation}\label{Lg_tdep}
\mathcal{L}g=\dfrac{\partial}{\partial \tau} (\bar{p}T)+w\cdot\nabla \Big( \bar{p}T\Big) +\mathcal{O}(\ep).
\end{equation}
Let us suppose again that $g$ is of the form
\[
g=-\dfrac{R(\tau,X,w)}{\mu(\tau,X,w)}+\alpha(\tau,X)T(\tau,X,w)\,
\]
where again $\alpha(\tau,X)$ does not depend on $w$. 
Inverting \eqref{Lg_tdep}, we find 
\[
g=-\dfrac{1}{\mu}\dfrac{\partial \bar{p}}{\partial \tau} T+\bar{p}\dfrac{\partial T}{\partial \tau}  +w\cdot\nabla \Big( \bar{p}T\Big)+\alpha(\tau,X) T(\tau,X,w)
\]
and then
\[
g=-\dfrac{1}{\mu}\Big[\Big(w T-{\bf E}_T T\Big)\cdot \nabla \bar{p}+\Big( w\cdot\nabla T-\nabla \cdot{\bf E}_T T+\dfrac{\partial T}{\partial \tau} \Big)\bar{p} \Big]+\alpha(\tau,X) T(\tau,X,w) .
\]
Since we assumed $g$ carries no mass and have $\dfrac{\partial}{\partial \tau}\bar{T}=0$, we obtain the same 
$\alpha$ as in \eqref{alpha_hyp}. Concluding, equation \eqref{macro.drift.3} becomes in this case
\begin{equation}\label{macro.drift.4}
\dfrac{\partial \bar{p}}{\partial \tau}+\nabla \cdot \big( \bar{p}({\bf E}_T +\epsilon a_h) \big)=\ep \nabla \cdot \nabla \cdot \Big( \mathbb{D}_h\bar{p}\Big)+\ep\nabla \cdot
\left(\int_V  (w-{\bf E}_T) {\bf E}_T \cdot \frac{\nabla T}{\mu} \, dw\; \bar p\right) +\ep \nabla \cdot \Big(\dfrac{\partial}{\partial \tau}{\bf E}_T \bar{p}\Big).
\end{equation}

\section{Diffusion equations for biological particles}\label{mod_turn_diff_eq}

To appreciate the power of the analysis, we use this section to consider special cases relevant to various biological scenarios. First, consider the parabolic limit equation (\ref{FAAD.xv}). Applications often only consider the macroscopic model for $\bar p_0$, and we simplify the notation by using $u(t,x)$ instead of the cumbersome $\bar p_0(\tau,X)$. We rewrite the parabolic limit (\ref{FAAD.xv}) in a more standard Fickian form 
\begin{equation}\label{FAAD.xv.new}
\dfrac{\partial u}{\partial t }(t,x)= \nabla \cdot \big[\mathbb{D}(x)\nabla u(t,x)\big] +\nabla \cdot \big[({\bf a}(x)+\nabla \cdot \mathbb{D}(x))u(t,x)\big].
\end{equation}
Therefore, the fully anisotropic advection-diffusion process described by the Fokker-Planck equation \eqref{FAAD.xv} is revealed as a standard anisotropic Fickian diffusion with an advective component given by the combination of drift velocity ${\bf a}(x)$ and the gradient of the diffusion tensor  $\nabla\cdot \mathbb{D}(x)$. We consider some special cases and simplify notation by omitting dependencies when they are obvious.  

\subsection{Dependences of turning rates and turning kernels}

\begin{enumerate}
    \item 
Suppose $\mu$ does not depend on $w$. The run time distribution, $T$, then becomes 
\[ T(x,w) = \frac{q(x,w)}{C(x)\mu(x)}, \quad \mbox{with }\quad  C(x) = \int_V \frac{q(x,w) }{\mu(x)} dw = \frac{1}{\mu(x)}\,. \]
Hence, in this case 
\[ T(x,w) = q(x,w)\]
and the model reduces to previously studied cases (e.g. see \cite{Hillen.05,Painter_Hillen.19}). The  
diffusion matrix here is given by the variance-covariance matrix of the fibre network distribution, $q(x,w)$, 
divided by the turning rate 
\[
\mathbb{D}(x) =\int_V w\otimes w\, \frac{T(x,w)}{\mu(x)} dw = \frac{1}{C(x)\mu(x)^2} 
\int_V w\otimes w q(x,w) dw  = \dfrac{\mathbb{V}_q(x)}{\mu(x)}.
\]
The drift velocity (\ref{macroap}), meanwhile, is
\begin{equation}\label{adw.new}
{\bf a}(x)=- \dfrac{\nabla \mu(x)}{\mu^2(x)}\mathbb{V}_q(x).
\end{equation}
In this case, if $q$ is even and then ${\bf E}_T$ vanishes, we obtain a parabolic limit equation (\ref{FAAD.xv}) as 
\begin{equation}\label{diff.3}
\dfrac{\partial u}{\partial t} - \nabla \cdot \left[\mathbb{V}_q \dfrac{\nabla\mu}{\mu^2} u  \right] = \nabla \cdot \nabla \cdot \left(\dfrac{\mathbb{V}_q}{\mu} u\right).
\end{equation}
This can again be written in Fickian form, where we obtain 
\begin{equation}\label{diff.3.fick}
\dfrac{\partial u}{\partial t}= \nabla \cdot \left(\dfrac{\mathbb{V}_q}{\mu}\nabla u \right)+\nabla \cdot \left[ \dfrac{\nabla \cdot \mathbb{V}_q}{\mu} u\right]\,.
\end{equation}

When $\mu$ is spatially heterogeneous, the diffusion process \eqref{diff.3} is a fully anisotropic process with advection given by the taxis term \eqref{adw.new}. This leads the dynamics towards decreasing values of $\mu$. The turning rate $\mu$ also scales the diffusivity, with large diffusivity for small turning rates and vice versa. Hence, frequently turning particles will not show large diffusive spread compared to those turning less often. The drift in the direction of the negative gradient of $\mu$ indicates a drift away from regions of low diffusion towards regions of high diffusion. 

\item If we further assume $\mu$ is constant and consider $q=q(x,w)$, the Fokker-Planck equation (\ref{FAAD.xv}) leads to
\begin{equation}\label{diff.2}
\dfrac{\partial u}{\partial t}= \frac{1}{\mu} \nabla \cdot \nabla \cdot(\mathbb{V}_q u).
\end{equation}
When the oriented habitat is not spatially homogeneous, \eqref{diff.2} describes a fully anisotropic process where diffusion is again linked to an anisotropic environment (which determines the direction of anisotropy) and to the frequency of reorientations (which determines the intensity of diffusion).

\item Suppose now $\mu$ is constant and further assume $q=q(w)$, i.e. the network is spatially homogeneous but can 
still describe an oriented environment. Then the diffusion equation \eqref{diff.2} becomes
\begin{equation}\label{diff.1}
\dfrac{\partial u}{\partial t} = \frac{1}{\mu} \nabla \cdot (\mathbb{V}_q \nabla u)\,. 
\end{equation}
Equation \eqref{diff.1} describes diffusion in a spatially homogeneous environment, where anisotropy is due to the presence of a dominant alignment in the environment. 

\item Suppose both turning rates and turning kernels depend on $w$ but not $x$, i.e. $\mu = \mu (w)$ and $q=q(w)$. If we further assume ${\bf E}_T=0$, in this case $\nabla\mu(X,w)=0$, the macroscopic drift from (\ref{macroap} ) satisfies ${\bf a}=0$ and the macroscopic diffusion (\ref{macroDp}) is 
\begin{equation}\label{macroDp2}
\mathbb{D} = \int_V w\otimes w \frac{T(X,w)}{\mu(w)} dw .
\end{equation}
We use this case later to analyse the movement patterns of cells migrating on fabricated anisotropic surfaces, such as the fibronectin stripe arrangements devised in \cite{doyle2009} (see the schematic in Figure \ref{doyleschema} and the simulations in Figure \ref{fig:Doyle}).

\item Finally, suppose a constant fibre distribution, i.e. 
\begin{equation}\label{q.cost}
q=1/2\pi,
\end{equation}
but $\mu = \mu (w,x,t)$, i.e. the turning rate can depend on the orientation of the environment. If we suppose ${\bf E}_T=0$, we are again in the previous case, but the eventual anisotropy will be the direction orthogonal to the dominant direction of $\mu$, as it is the orientation of cells that tend to turn more frequently.
\end{enumerate}

\subsection{Isotropy and anisotropy}\label{iso.an}

Generally we observe an anisotropic diffusion process (\ref{macroDp}), a consequence of both cells orienting with respect to the network alignment via $q(x,w)$ and the direction-dependent turning frequency $\mu(x,w)$. Specifically, diffusive anisotropy is encapsulated through the second moment of the distribution 
\[ \frac{T(x,w)}{\mu(x,w)} = \frac{q(x,w)}{C(x) \mu^2(x,w)}. \]
The nominator/denominator localisations of $q$ and $\mu$ naturally suggest that turning direction and turning rate might have opposing impact on macroscopic dynamics. 
\begin{enumerate}
\item For example, suppose that 
\begin{equation}\label{q.mu2}
q(x,w) \sim \mu^2(x,w)\,.
\end{equation}
Here we have $T(x,w)\sim \mu(x,w)$ and obtain a macroscopic drift term of the form ${\bf E}_T(x) \sim  \int_V w\mu(x,w)  dw.$ The macroscopic diffusion tensor (\ref{macroDp}) in this case will be isotropic 
\[
\mathbb{D} \sim \displaystyle \int_V  w\otimes w  dw = c I. 
\]
Hence, for $\int w\mu(x,w)dw \neq 0$,  the macroscopic regime will be isotropic and drift-dominated. 

\item Suppose instead
\begin{equation}\label{mu.q}
q(x,w) \sim \mu(x,w)\,.
\end{equation} 
Here, $T(x,w)$ is constant and the macroscopic regime will always be diffusive with
${\bf E}_T=0$. The diffusion might be anisotropic, since
\[ \mathbb{D}\sim \displaystyle \int_V w\otimes w \dfrac{1}{q(x,w)} dw.
\]
\item It is also relevant to consider the case in which 
\begin{equation}\label{mu.invq}
q(x,w) \sim \frac{1}{\mu(x,w)}\,,
\end{equation} 
because in many biological cases cells are likely to both choose the fibre direction and turn less frequently when moving in the direction of dominating alignment. In this case we have both a drift-driven phenomenon, with ${\bf E}_T=\int w\frac{1}{\mu^2(x,w)}dw$, and the diffusion can be anisotropic according to
\[ \mathbb{D}\sim \displaystyle \int_V w\otimes w q^3(x,w) dw.
\]
\end{enumerate}

 \subsection{Comparison to diffusion of passive particles} 


 The somewhat curious diffusive terms derived here arise due to the study of {\em active movers}: biological particles, such as cells and organisms, generate their own movement energy and are therefore unconstrained by energy and momentum conservation as demanded by classical physics. For passive movers (movers simply transported by a fluid environment or some other stream) the situation is different: in that context, an equation of the form  \eqref{diff.1} corresponds to the equation describing diffusion in an anisotropic but homogeneous environment, i.e. without spatial variation; equation \eqref{diff.2} corresponds to the equation derived in \cite{Chapman}, describing particle diffusion in a heterogeneous environment; equation \eqref{diff.3} corresponds to the case in which $\mu$ depends on the position, and it is then the only case in which the macroscopic equation has that particular mixed structure. Within the diffusion theory of physical particles, this mixed structure arises when describing a so-called thermal effect, $\ie$ when there is a temperature gradient \cite{Wereide}. Here, the variable turning rate can be viewed somewhat analogously to temperature in a system of physical particles\footnote{In a biological context, this `temperature' should not be thought of in its everyday sense, but the macroscopic temperature that arises from the movement of particles viewed as a mechanical multi-particle system, like in the thermodynamic theory of gases.}, and the two factors influencing the dynamics are the spatial heterogeneity and the adaptability to environmental heterogeneity.

 An example of adaptability in a biological context is provided in \cite{Kim.18, Kim.2013.1}, addressing starvation-driven diffusion. Here, starvation induces the organism to increase motility and find a better environment, even if it is not known where that may be, $\ie$ the migration is unbiased. We further note that the three forms of diffusion equation  \eqref{diff.3},\eqref{diff.2},\eqref{diff.1} can be derived from space-jump processes, where variation in jumping depends on environment assessment of the current location, between locations or the destination  \cite{Othmer_Stevens.97,Kim.18,LutscherHillen}, with spatially homogeneous jumping times. When both the jumping time and length are spatially heterogeneous, a mixed structure similar to \eqref{FAAD.xv} arises.

\section{Numerical examples}\label{mod_turn_num_tests}

We present simulations to illustrate key model features,  numerically integrating the kinetic transport equation (\ref{transport.general}),\eqref{operator.2} to approximate the density distribution $p$ and in turn the corresponding macroscopic density via \eqref{rho}. For computational convenience we restrict to a 2D spatial setting, a rectangular 2D region $\Omega = [0,L_x] \times [0,L_y]$, and restrict the velocity space to $V=\mathbb{S}^1$, so that $w=\hat{w}$ and the speed is set at unit value. In particular, we integrate Eq. \eqref{transport.general},\eqref{operator.2} as in \cite{Filbet}, with the only difference lying in integrating the relaxation step semi-implicitly, as $\bar{p}_{\mu}$ has to be computed from the $p$ obtained from the current time step. 

Initially the population macroscopic density, $\bar{p}_0(x)=\bar{p}(0,x)$, is described by a tightly-concentrated Gaussian distribution centred at location $(x_0, y_0)$, e.g. see Figure \ref{Test1}(a), while cell orientations are uniformly distributed over $\mathbb{S}^1$. At the boundaries we set diffusive boundary conditions \cite{Lods}, yielding no-flux boundary conditions at the macroscopic level \cite{Plaza,loy2019JMB} both in the parabolic and hyperbolic limit.

The remainder of this section is divided into three core tests, designed to illustrate how $q$ and $\mu$ alter the macroscopic dynamics. Test 1 explores the extent to which anisotropic or isotropic behaviour emerges with the relationship between $q$ and $\mu$. In Test 2 we demonstrate the tactic effect induced by the spatial gradient of $\mu$. 
Test 3 extends the analysis to more complicated network structures.

\subsection{Arrangements}

We specify $q$ and $\mu$ using bimodal von Mises distributions, a standard circular distribution with known analytical forms for the first and second moments (e.g. see \cite{HMPS}). For notational convenience we introduce the following short-hand for a bimodal von-Mises distribution, with a given concentration parameter $k>0$ and a given unit vector $\theta$:
\begin{equation}\label{bvm} 
\bvm (w) = \frac{1}{4\pi I_0(k)} \left(e^{k w\cdot\theta} + e^{-k w\cdot \theta} \right).
\end{equation}
To further simplify notation we identify a unit vector by its angle, i.e. writing $\theta= (\cos\theta, \sin\theta)^T$. 

Fibres are arranged in two principal patterns, schematised in Fig. \ref{cross}. We base $q$ and $\mu$ on \eqref{bvm}, stipulating functions $\theta_q$ and $k_q$ for $q$ and $\theta_\mu$ and $k_\mu$ for $\mu$. $k_q = 0$ therefore corresponds to unaligned fibres and large $k_q$ generates highly aligned fibres along the axis $\theta_q, \theta_q+\pi$. Tests are devised as follows.
\begin{itemize}
\item {\bf Test 1}. Here we set $\theta_q = \mu_q = \pi/4$, $k_{q}=50e^{-0.25((x-5)^2+(y-5)^2)}$ and $\mu$ as one of $\mu \sim constant, \sim q, \sim q^2, \sim 1/q.$
\item{\bf Test 2}. We set $\theta_q$, $\mu_q$ and $k_q$ as in Test 1, but shift $k_{\mu}$ to $50e^{-0.25((x-4.5)^2+(y-4.5)^2)}$.
\item {\bf Test 3}. Setting $\Omega_{XY}=\Omega_X \cap \Omega_Y$, where $\Omega_X=\lbrace (x,y): 4 \leq x \leq 6\rbrace$, $\Omega_Y \lbrace (x,y): 4\leq y \leq 6 \rbrace$, we consider
\begin{equation}\label{q.cross}
q(x,w)\!\!=\!\!
\begin{cases}
\dfrac{1}{2\pi} \quad \rm{on} \quad \Omega-\Omega_{XY},\\[8pt]
\mbox{\sf bvM}_{k_q,\theta_q} \quad \rm{on} \quad \Omega_X,\\[8pt]
\mbox{\sf bvM}_{k_q,\theta_q^{\bot}} \quad \rm{on} \quad \Omega_Y,\\[8pt]
\dfrac{1}{2}\Big(\mbox{\sf bvM}_{k_q,\theta_q}+\mbox{\sf bvM}_{k_q,\theta_q^{\bot}}\Big) \, \rm{on} \, \Omega_{XY},
\end{cases}\!\!\!\!\!\!\!\!\mu(x,w)\!\!=\!\!\begin{cases}
\dfrac{1}{2\pi} \quad \rm{on} \quad \Omega-\Omega_{XY},\\[8pt]
\mbox{\sf bvM}_{k_{\mu},\theta_\mu} \quad \rm{on} \quad \Omega_X,\\[8pt]
\mbox{\sf bvM}_{k_{\mu},\theta_\mu^{\bot}} \quad \rm{on} \quad \Omega_Y,\\[8pt]
\dfrac{1}{2}\Big(\mbox{\sf bvM}_{k_{\mu},\theta_{\mu}}+\mbox{\sf bvM}_{k_{\mu},\theta_{\mu}^{\bot}}\Big) \, \rm{on} \, \Omega_{XY}.
\end{cases}
\end{equation}
Specifically, we will set $\theta_q = \pi$ and $k_q = 50$ to form the cross configuration on the right of Fig. \ref{cross}, where away from the cross fibres are isotropic and on the horizontal (vertical) arm fibres are aligned in the horizontal (vertical) direction. At the centre, fibres cross. For $\mu$ we consider related but subtly distinct forms, allowing distinct turning rate to turning distribution behaviour.
\end{itemize}

\begin{figure}[!t]
    \centering
    \includegraphics[width=0.8\textwidth]{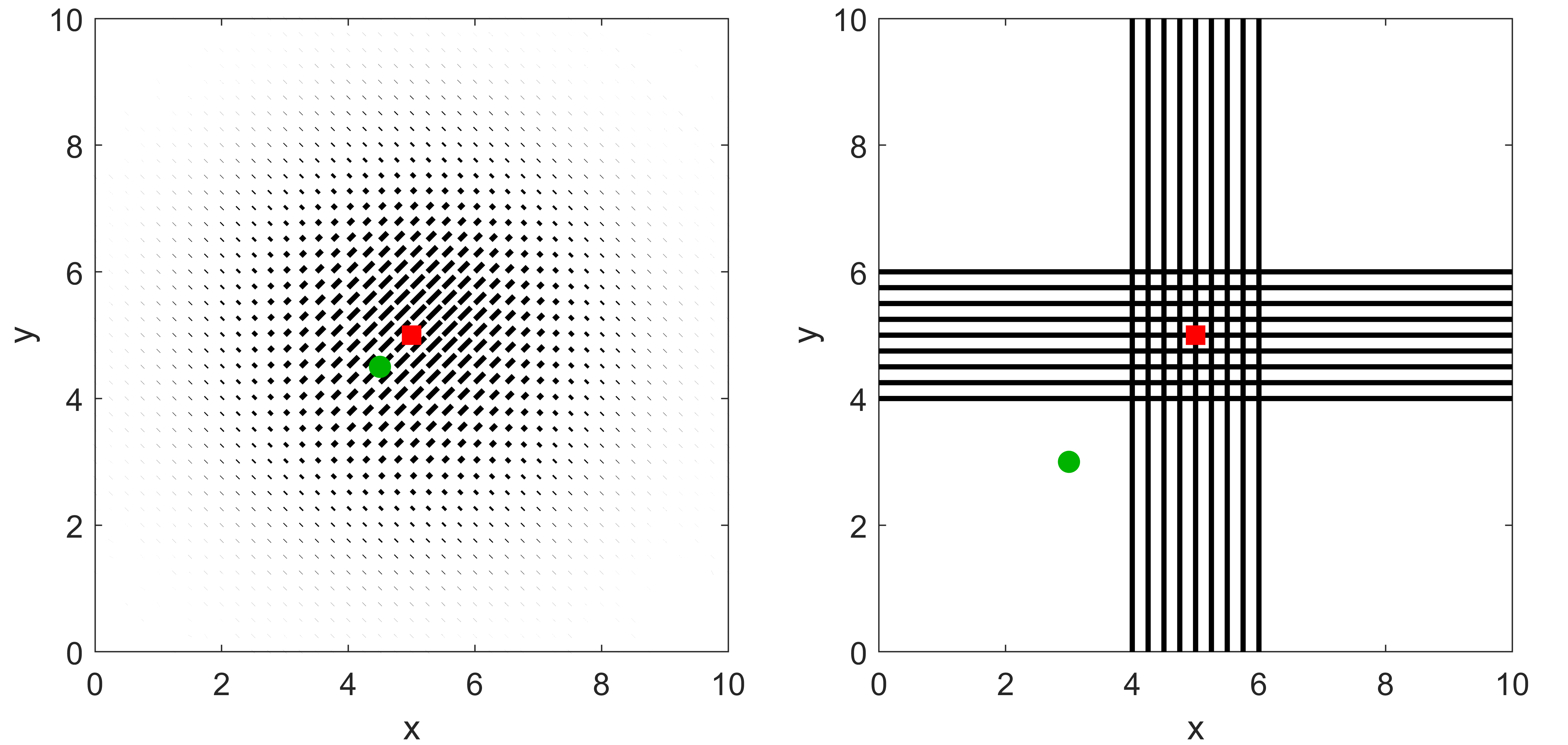}\hspace*{5mm}
    \caption{Left: Arrangement and intensity of fibre orientation ($q$) for {\bf Tests 1 \& 2}. The maximum of $\mu$ is positioned at the red square for Test 1 and the green circle for Test 2. Right: Arrangement and intensity of fibre orientation ($q$) for {\bf Test 3}. According to the precise simulation, we again shift the position of the maximum of $\mu$ from the red square to the green circle.}
    \label{cross}
\end{figure}

\subsection{Simulations}

{\bf Test 1} was designed to explore the extent to which isotropic or anisotropic behaviour arises with the choices of $q$ and $\mu$, for example related as in Eq. \eqref{q.mu2}, \eqref{mu.q} or \eqref{mu.invq}. Initialising according to Fig. \ref{Test1}(a), we plot the macroscopic density at $T=7.5$ in Fig. \ref{Test1}(b-e) for (b) $\mu \sim constant = 1$, (c) $\mu \sim q$, (d) $\mu \sim q^2$, (e) $\mu \sim 1/q$. In (f) we set $q$ constant but maintain anisotropy through $k_\mu = 50e^{-0.25((x-5)^2+(y-5)^2)}$ and $\theta_\mu=\pi/4$.

In particular, for Fig. \ref{Test1}(b) we recover the original model of \cite{Hillen.05} and observe anisotropic diffusion according to the dominating fibre alignment. In Fig. \ref{Test1}(c) we set $\mu=q$, i.e. situation (\ref{mu.q}). There is no drift, but anisotropy arises through the directional bias in $q$. This generates a conflict in which cells preferentially choose the dominating fibre alignment but, when facing those directions, turn more frequently. Consequently the anisotropy becomes orthogonal to the dominating fibre alignment. In Fig. \ref{Test1}(d) we consider case (\ref{q.mu2}) by choosing $\mu=\sqrt{q}$. Again,
preferential movement along an axis is countered by increased turning, but the weighting now generates
isotropic diffusion. Further, there is a non-trivial drift ${\bf E}_T$. The dynamics are similar to those in Figure \ref{Test1}(c), yet the pattern is more symmetric due to the isotropic diffusion. In Fig. \ref{Test1}(e) we consider case (\ref{mu.invq}), i.e. where $\mu=1/q$. Cells now turn less frequently along the directions of preferential alignment and, intuitively, we should expect enhanced anisotropic diffusion along certain axes. Simulations confirm this, with the cells remaining even more tightly aligned along the dominating fibre orientation as compared to Fig. \ref{Test1}(b). Finally, in Fig. \ref{Test1}(f) the fibre network does not impact on orientation, but cells oriented along the axis ($\pi/4, 5\pi/4$) turn more frequently. Diffusion remains anisotropic as predicted by \eqref{q.cost}.

\begin{figure}[!t]
\centering
\subfigure[]{\includegraphics[scale=0.07]{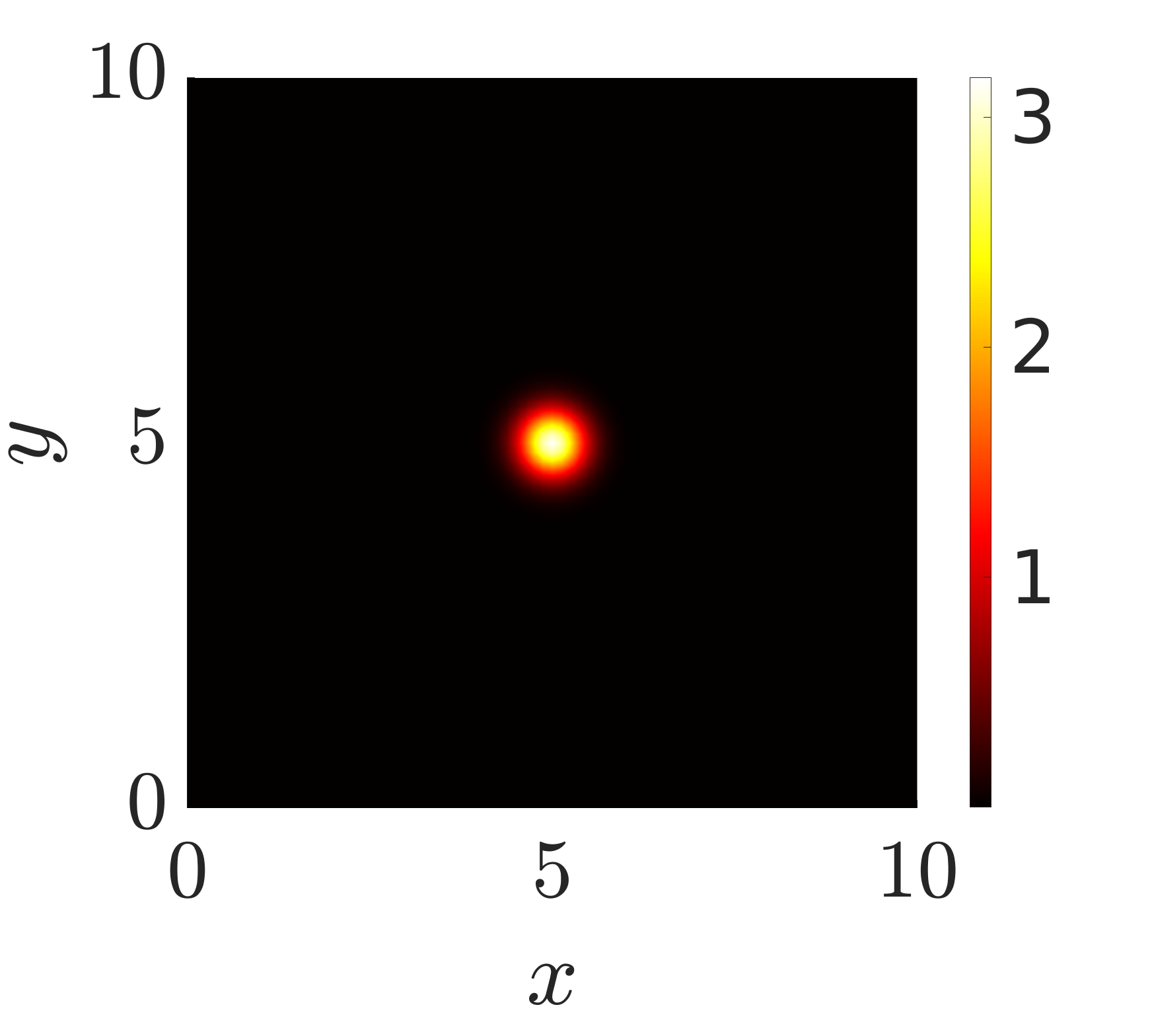}}
\subfigure[]{\includegraphics[scale=0.11]{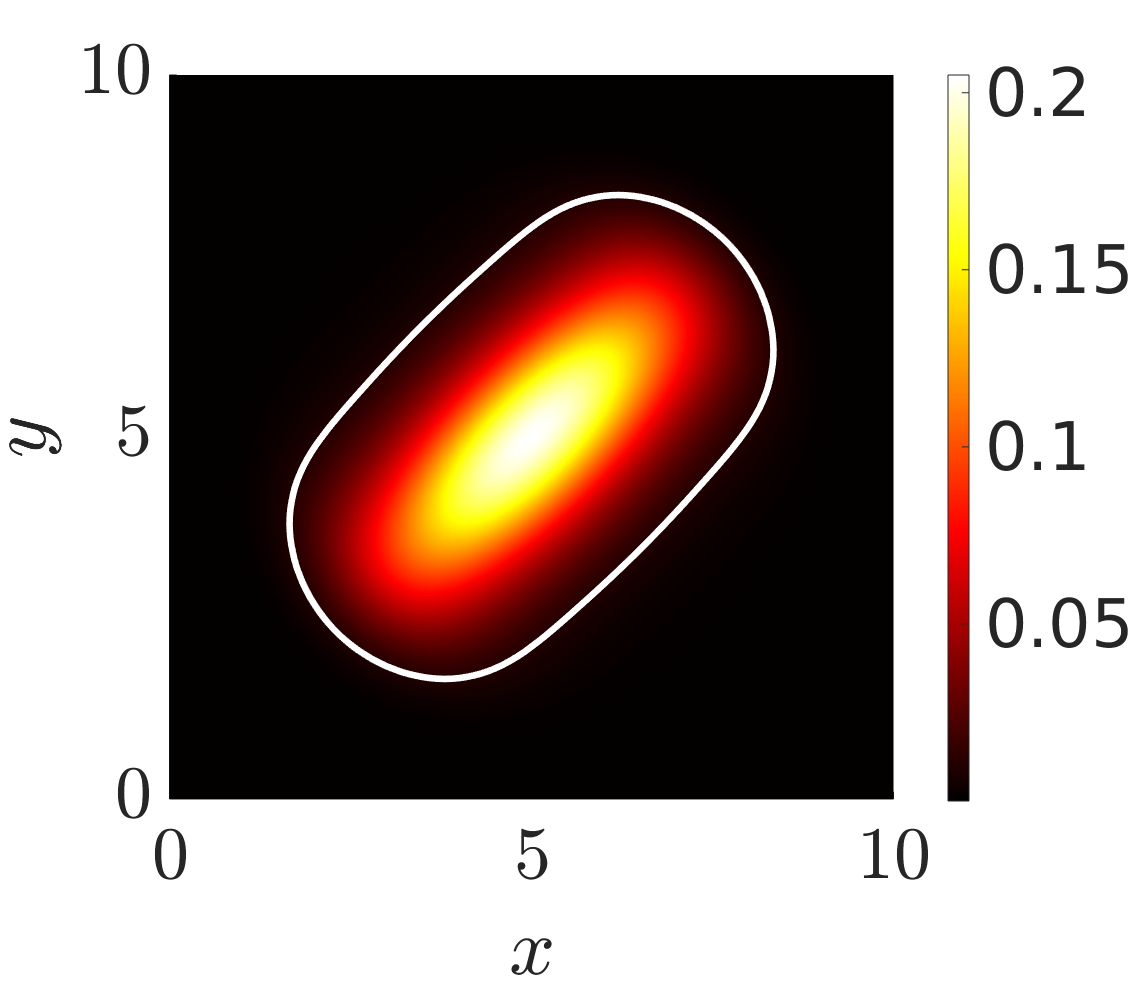}}
\subfigure[]{\includegraphics[scale=0.11]{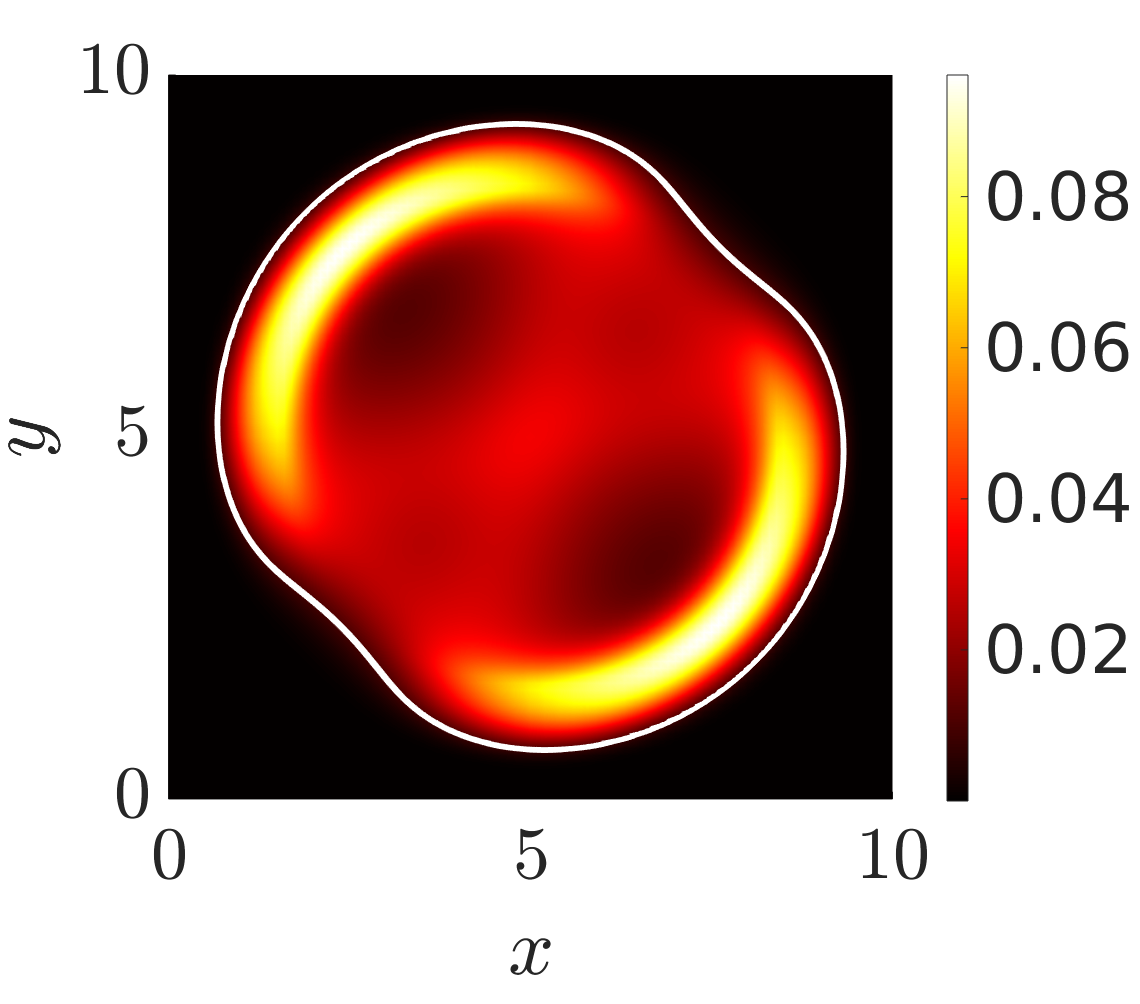}}\\
\subfigure[]{\includegraphics[scale=0.11]{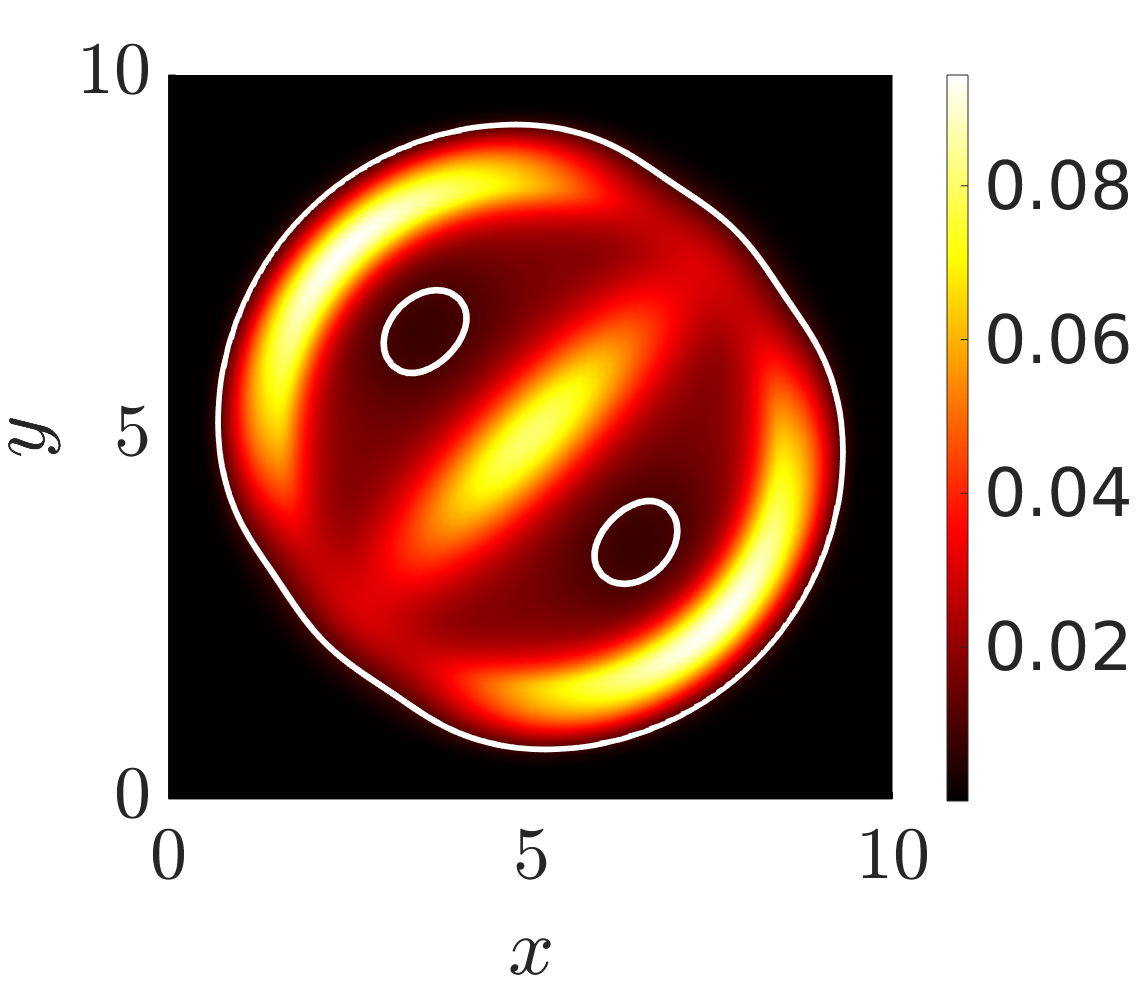}}
\subfigure[]{\includegraphics[scale=0.11]{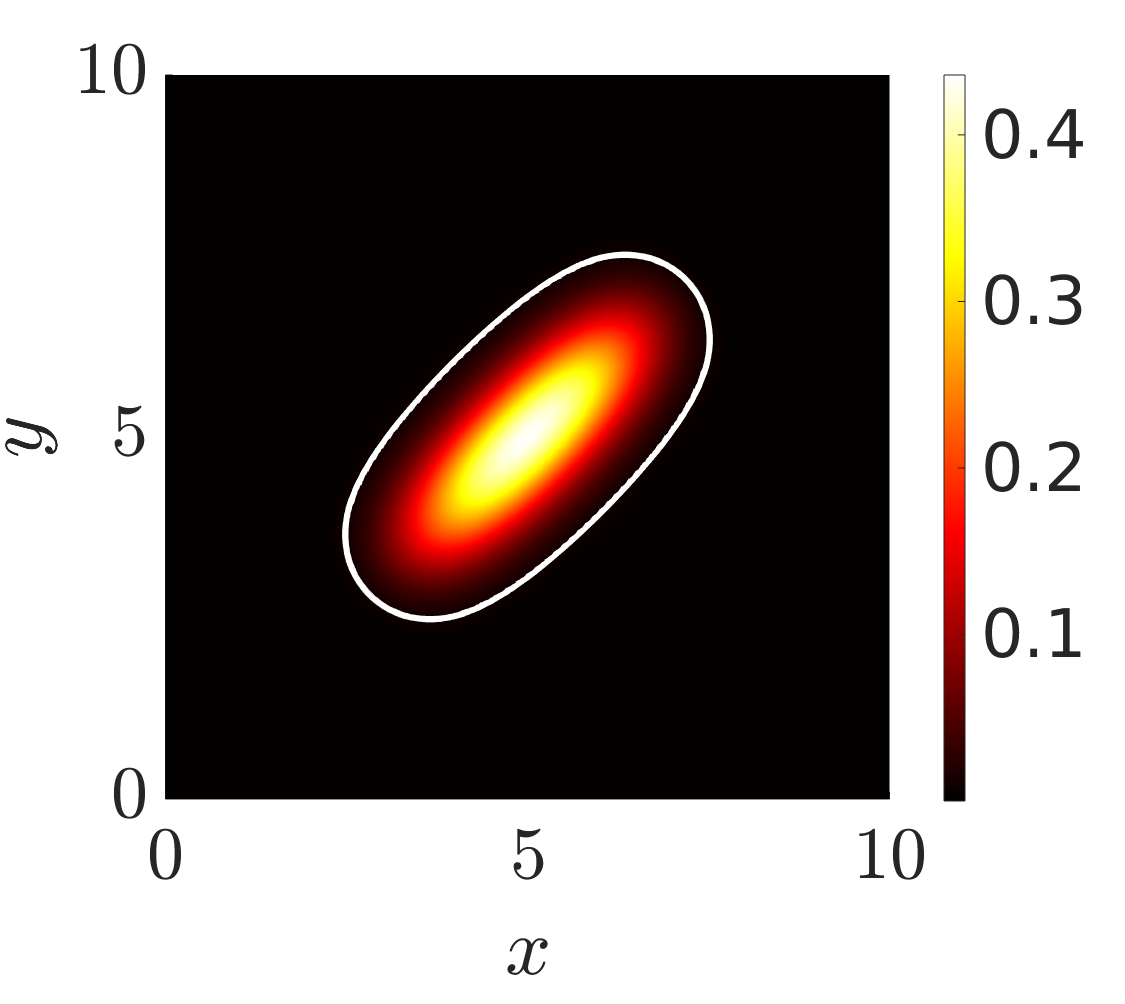}}
\subfigure[]{\includegraphics[scale=0.11]{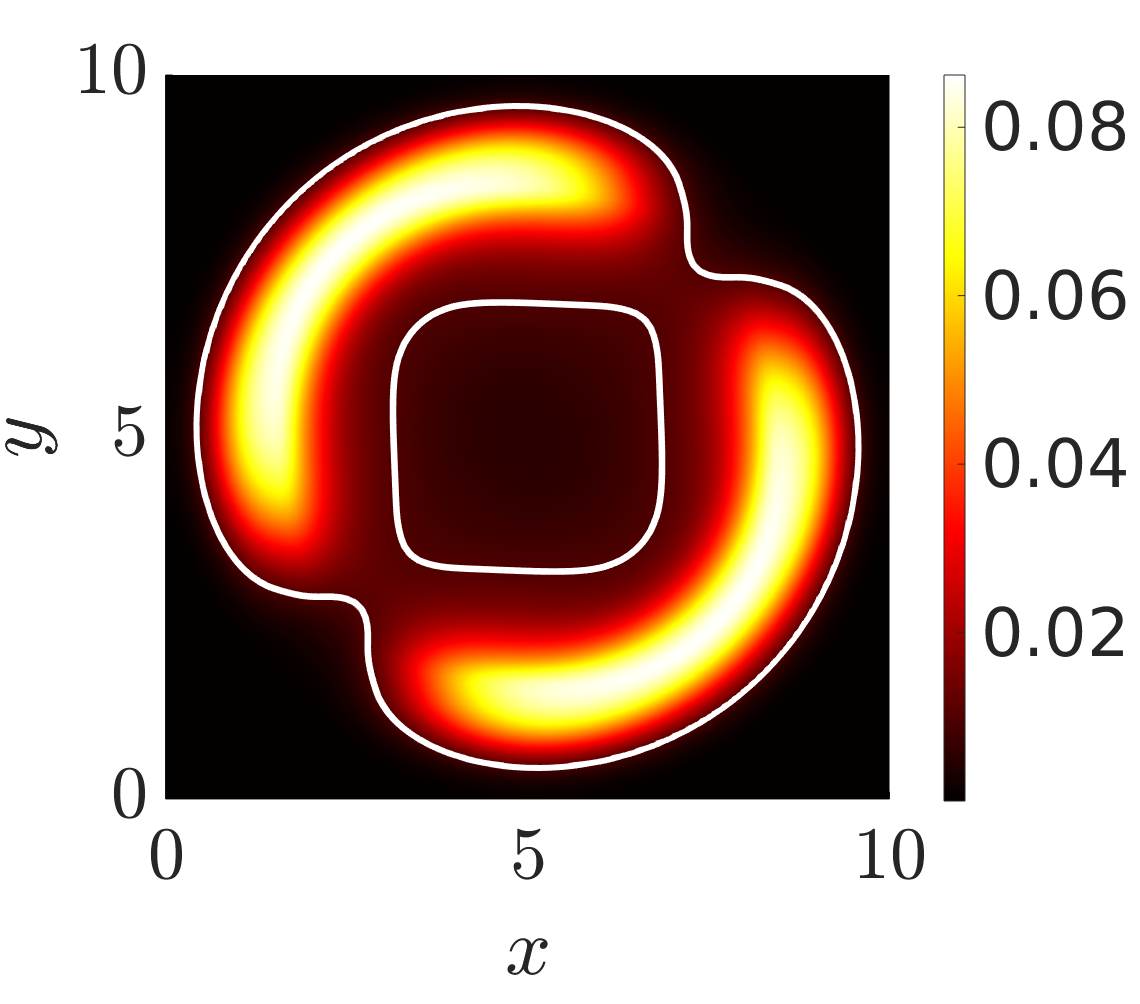}}
\caption{
{\bf Test 1}. Anisotropic/isotropic spread according to the choices of $q$ and $\mu$. See text for details. (a) Initial macroscopic density $\bar{p}(0,x)$, (b-f) macroscopic density $\bar{p}(t,x)$ at $t=7.5$. The white curve in (b-f) is the level set defined by $\bar{p}(t,x)=0.1$. In (b-e) $q$ is as described in the text and the turning rate is given by: (b) $\mu=1$, (c) $\mu=q$, (d) $\mu=\sqrt{q}$, (e) $\mu=1/q$. In (f) $q=1/2\pi$ while $\mu$ is given by \eqref{bvm}, with $\theta_{\mu}=\pi/4$ and $k_{\mu}=50e^{-0.25((x-5)^2+(y-5)^2)}$. }
\label{Test1}
\end{figure}

{\bf Test 2} was designed to show the taxis induced by the spatial variability of the turning rate $\mu$ 
and simulation results are reported in Fig. \ref{Test1} (c). Specifically, we shift the peak of the turning rate away from the centre of the domain, to the point represented by the green dot. There is a subsequent tendency of cells to avoid this new location, coherent with equation \eqref{adw.new} indicating greater diffusion with higher values of the turning frequency (Figure \ref{fig.2}).
\begin{figure}[!t]
\centering
\subfigure[$t=2.5$]{\includegraphics[scale=0.1]{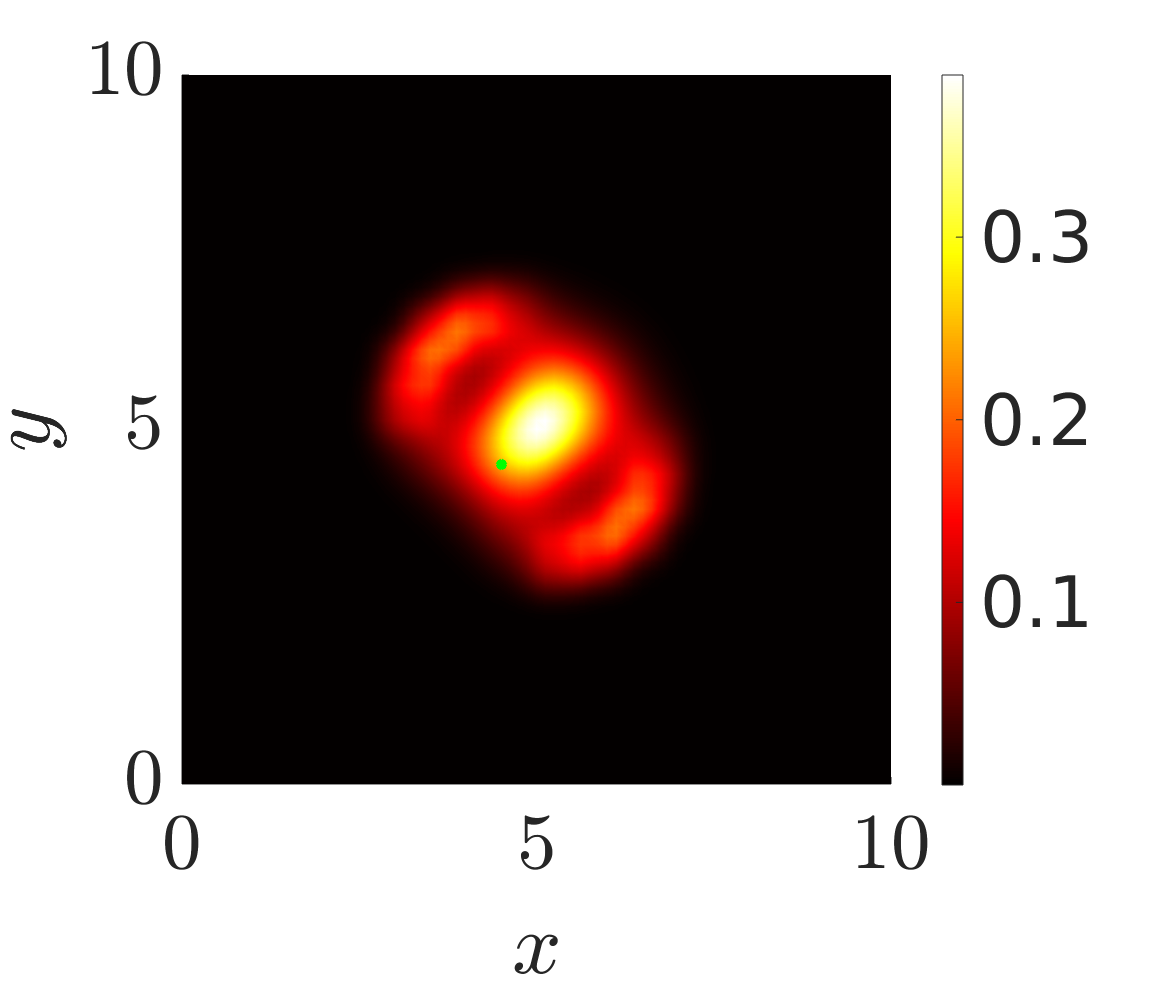}}
\subfigure[$t=5$]{\includegraphics[scale=0.1]{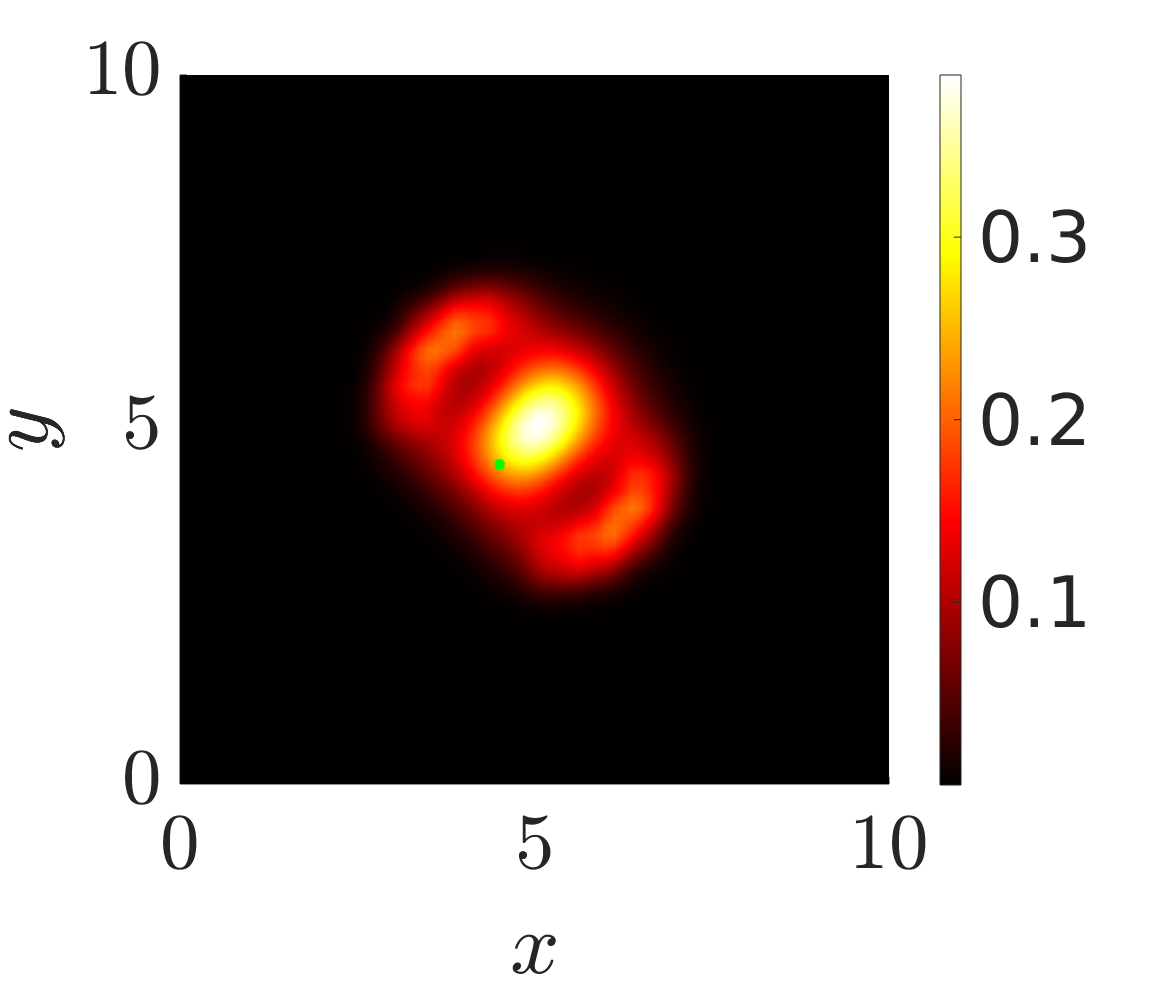}}
\subfigure[$t=7.5$]{\includegraphics[scale=0.1]{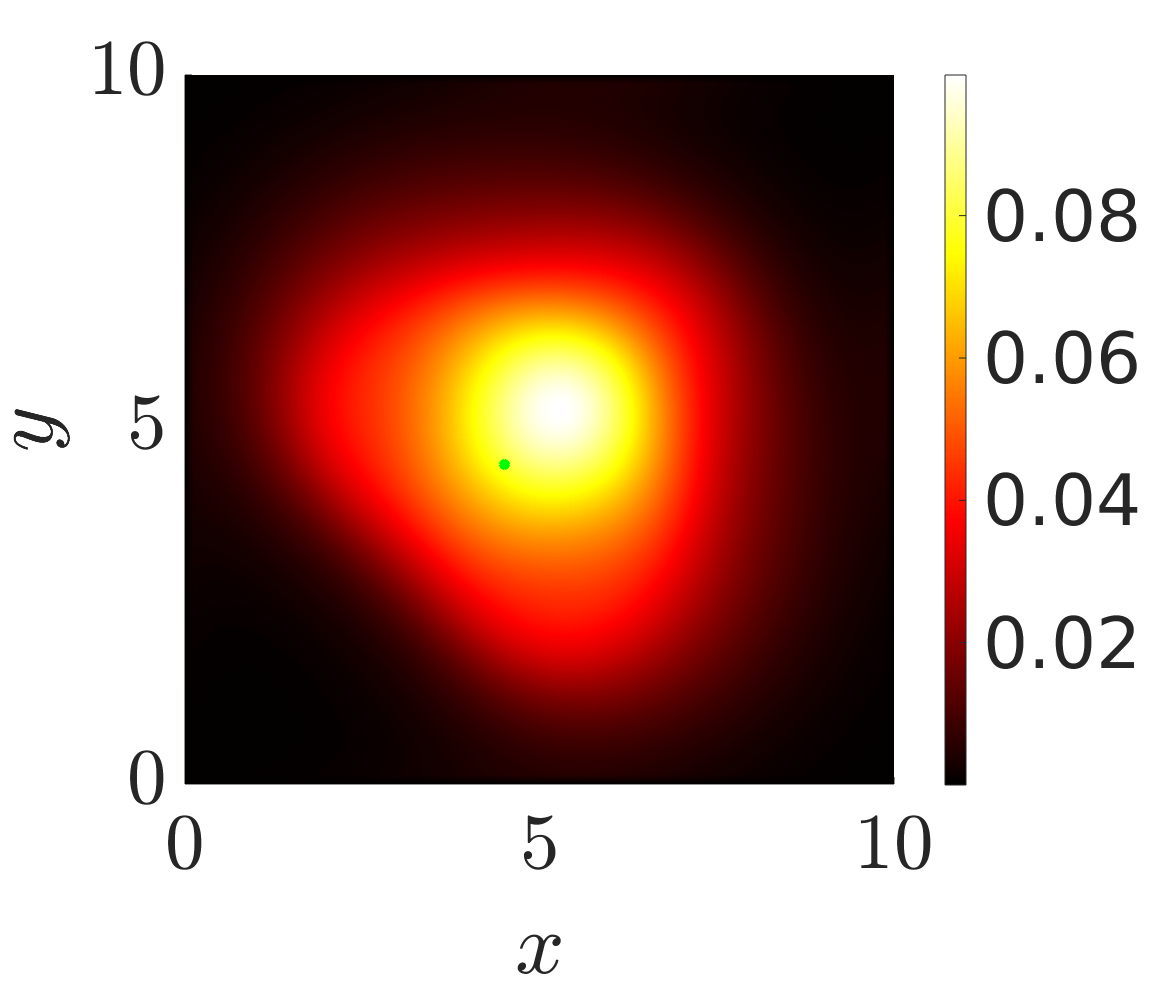}}
\caption{{\bf Test 2.} Taxis induced through variable turning rates. $q$ is as described in the left of Fig. \ref{cross} (see text for details), while the peak of the turning frequency $\mu$ is shifted to centre on the green dot. Simulations plot the macroscopic density at successive times $t = 2.5, 5, 7.5$.}
\label{fig.2}
\end{figure}

{\bf Test 3} extends these analyses to a more complicated environment, see Fig. \ref{cross}. Fig. \ref{fig.4}(a) plots the initial (macroscopic) cell distribution for the experiments (b) and (c), while Fig. \ref{fig.4}(d) shows the distribution used for (e) and (f). In Fig. \ref{fig.4}(b), we set $\mu \sim q$ so that cells orient and migrate along fibres but this is counterbalanced by turning more frequently when moving in those directions. Spread is subsequently inhibited by the cross arrangement. In Fig. \ref{fig.4}(c) we consider instead $\theta_\mu=\theta_q^{\bot}$, so that particles both follow the fibres and turn less frequently when moving in their direction. As expected, there is a clear tendency of cells to follow the cross structure. The second row performs the same simulations, but relocating the initial cell distribution (now centred at $(4,4)$) and the peak of $k_\mu$ to $(3,3)$. The latter generates an advection away from this point, a consequence of the decreasing gradient of $\mu$ experienced by cells. Fig. \ref{fig.4} (e) shows even more clearly the inhibition resulting from fibres at the centre of the cross, while in (f) cells are seen to spread rapidly along the arms once it has been reached. Note the different time scales between Fig. \ref{fig.4}(c) and Fig. \ref{fig.4}(f). 

\begin{figure}[!t]
\centering
\subfigure[$t=0$]{\includegraphics[scale=0.07]{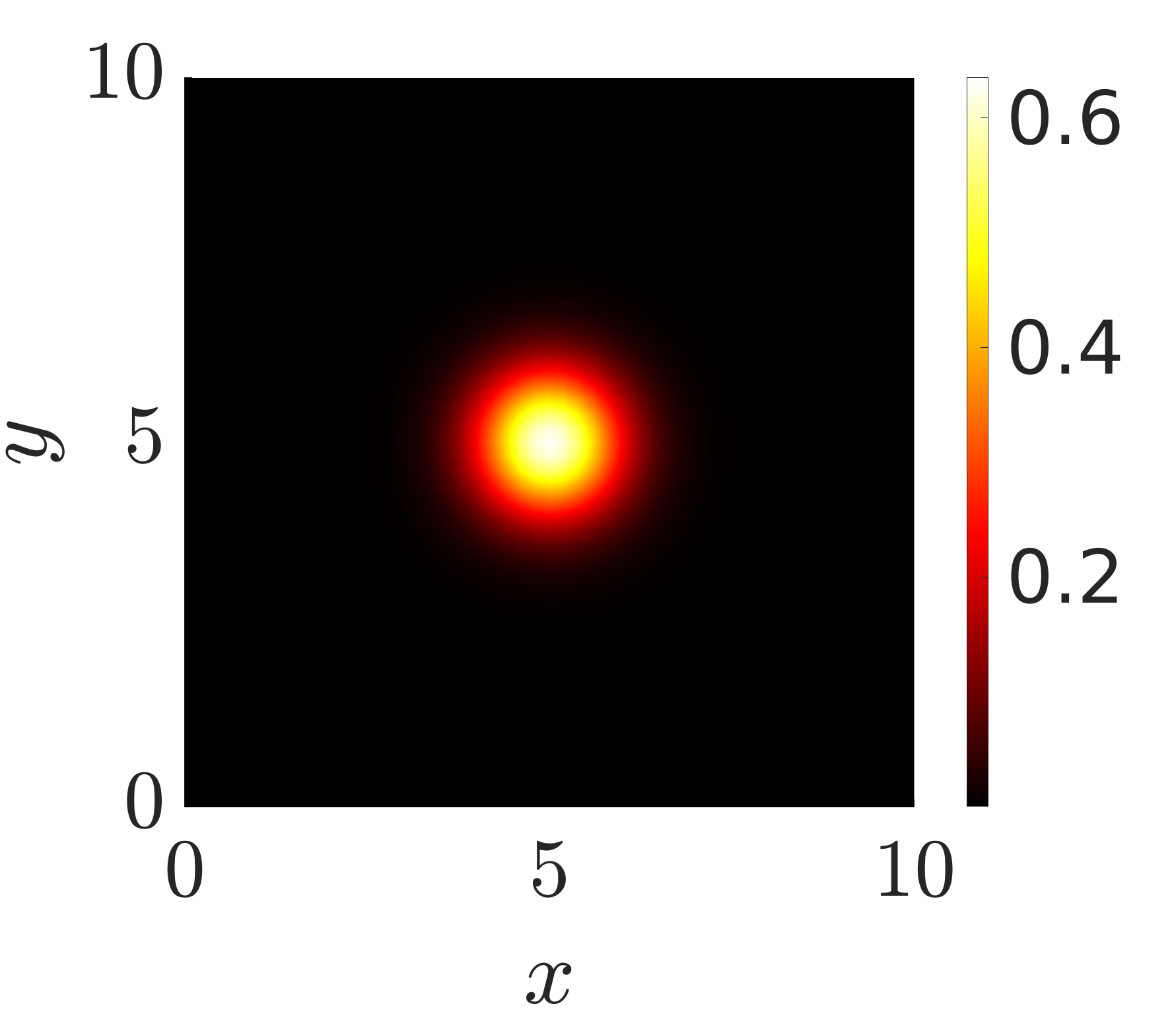}}
\subfigure[$t=7.5$]{\includegraphics[scale=0.07]{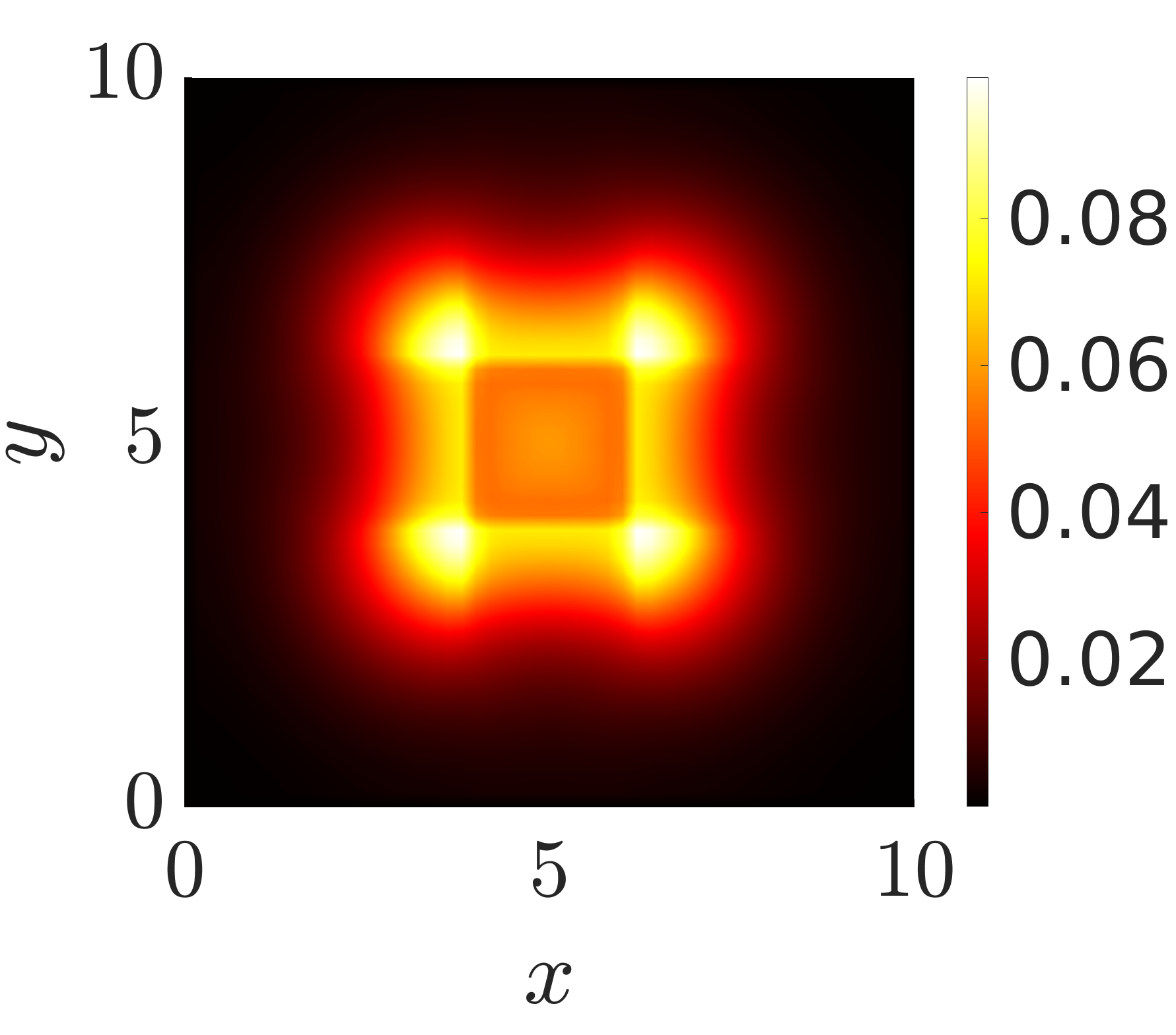}}
\subfigure[$t=15$]{\includegraphics[scale=0.07]{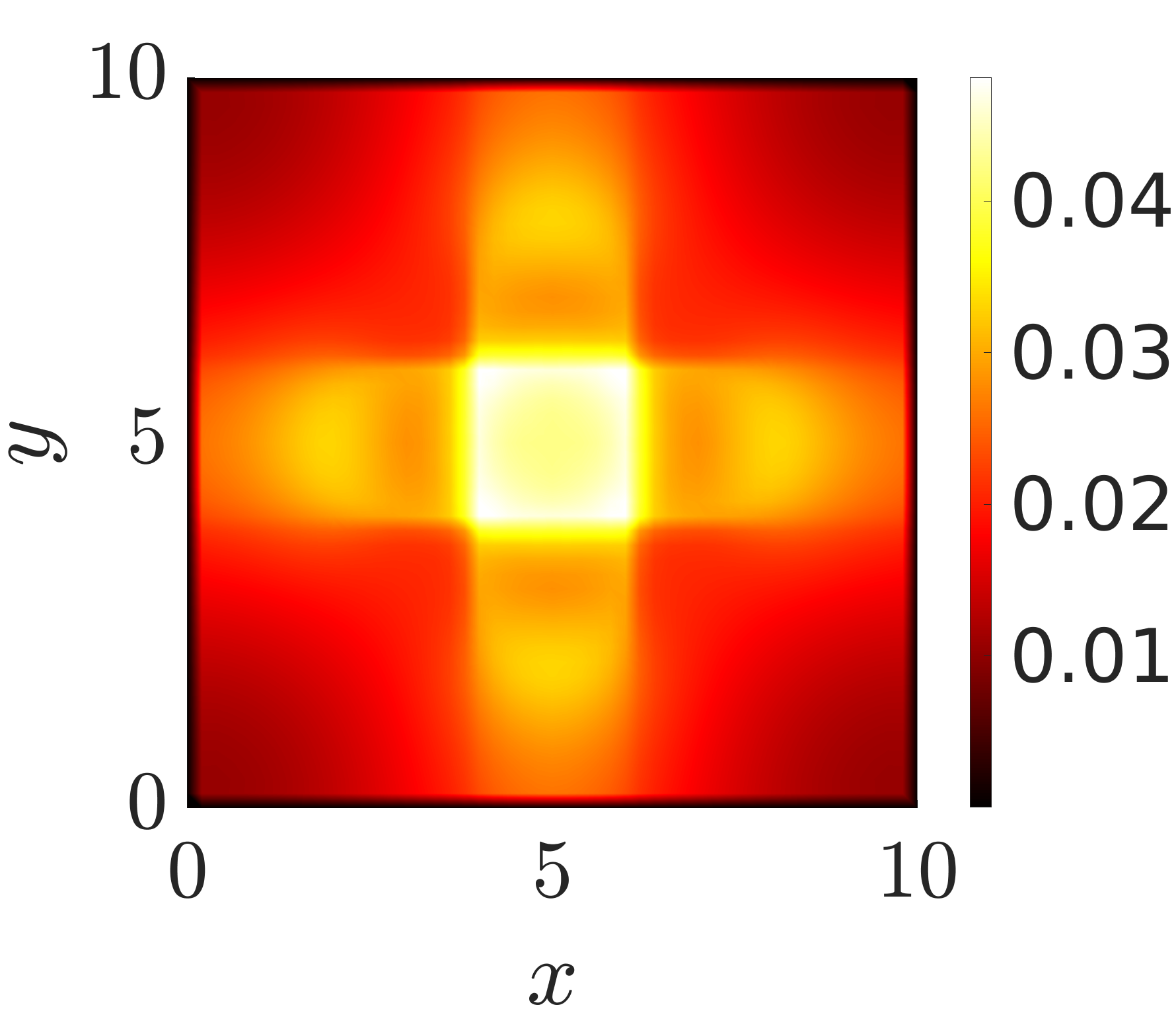}}\\
\subfigure[$t=0$]{\includegraphics[scale=0.07]{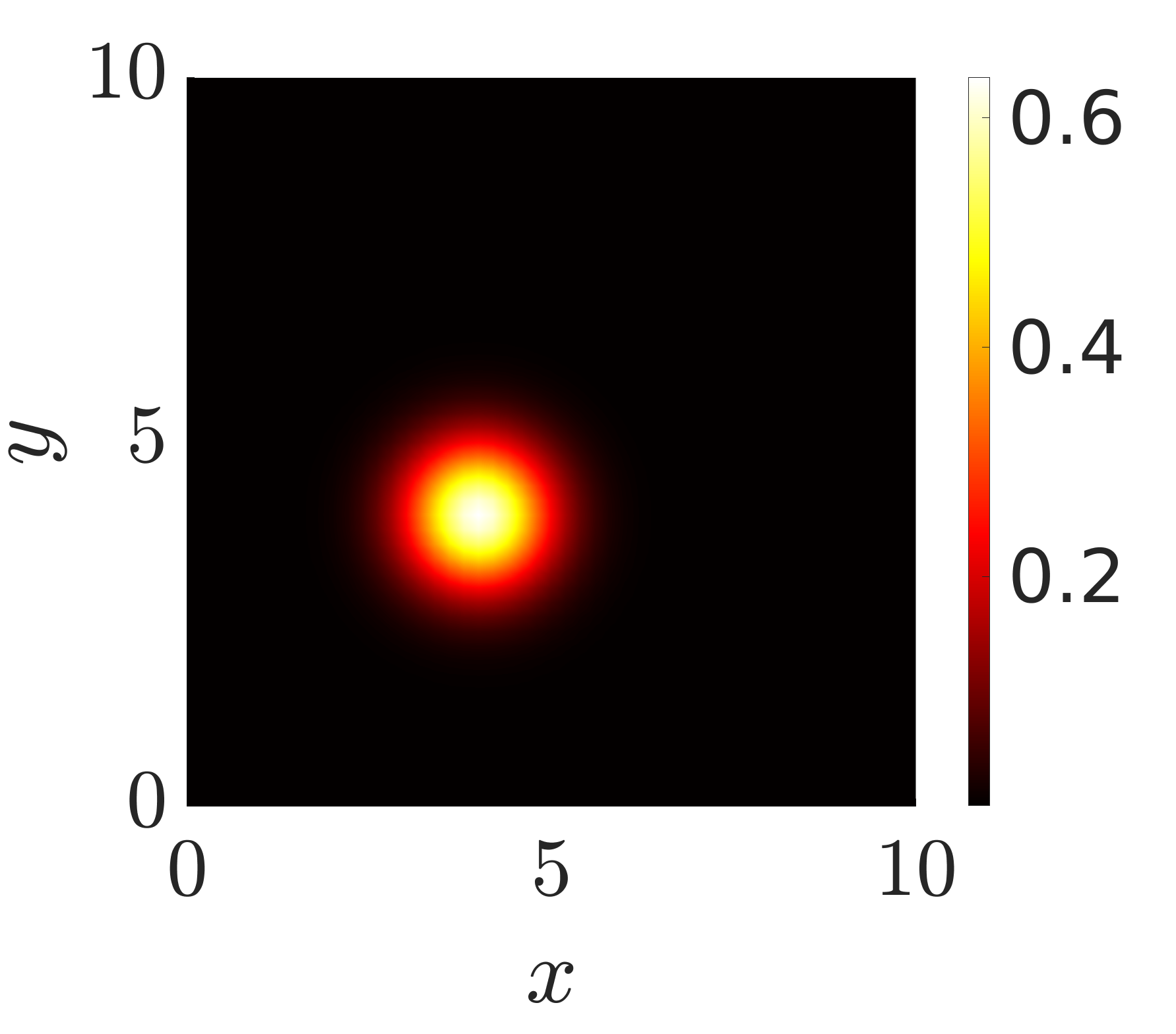}}
\subfigure[$t=7.5$]{\includegraphics[scale=0.07]{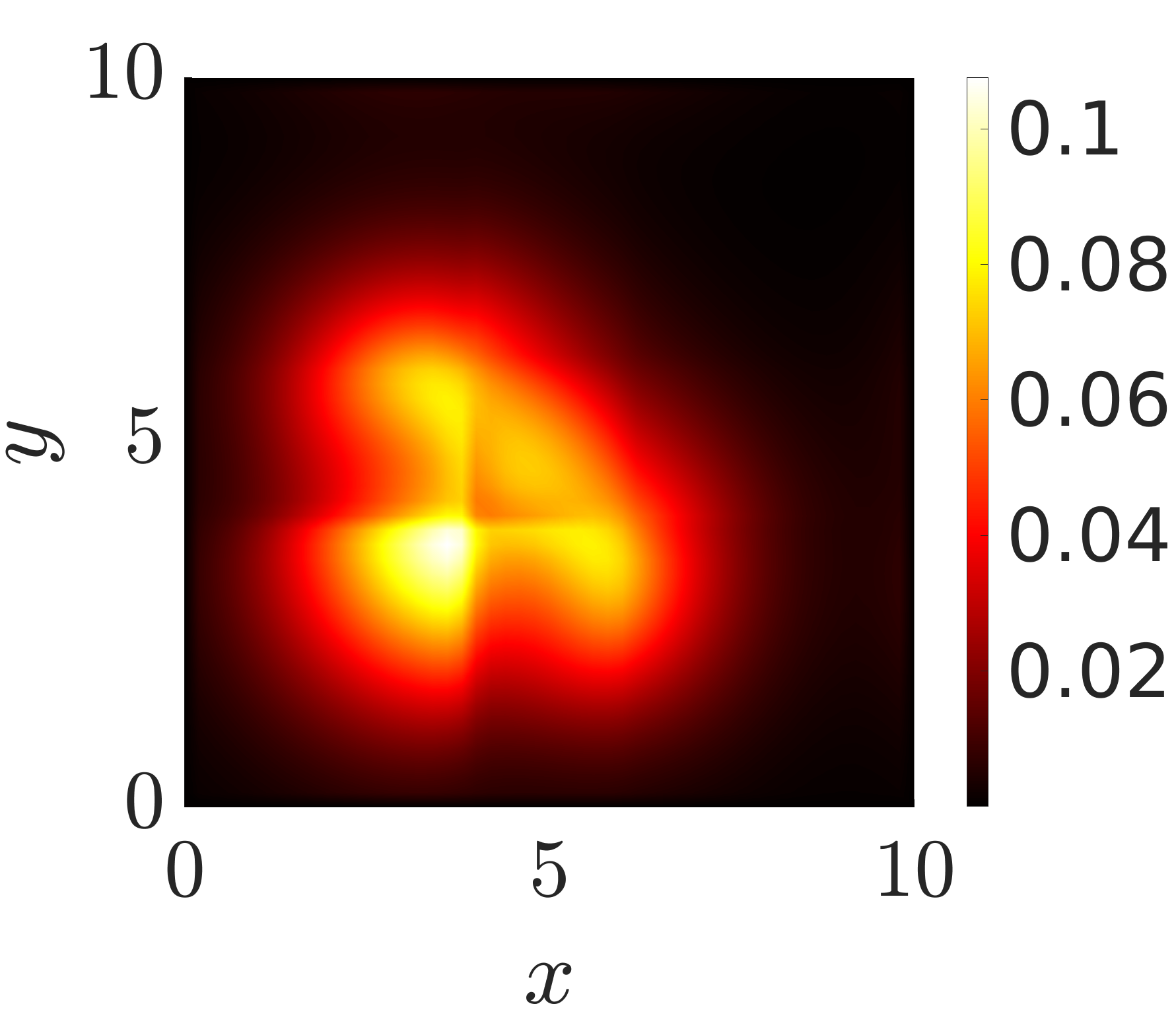}}
\subfigure[$t=7.5$]{\includegraphics[scale=0.07]{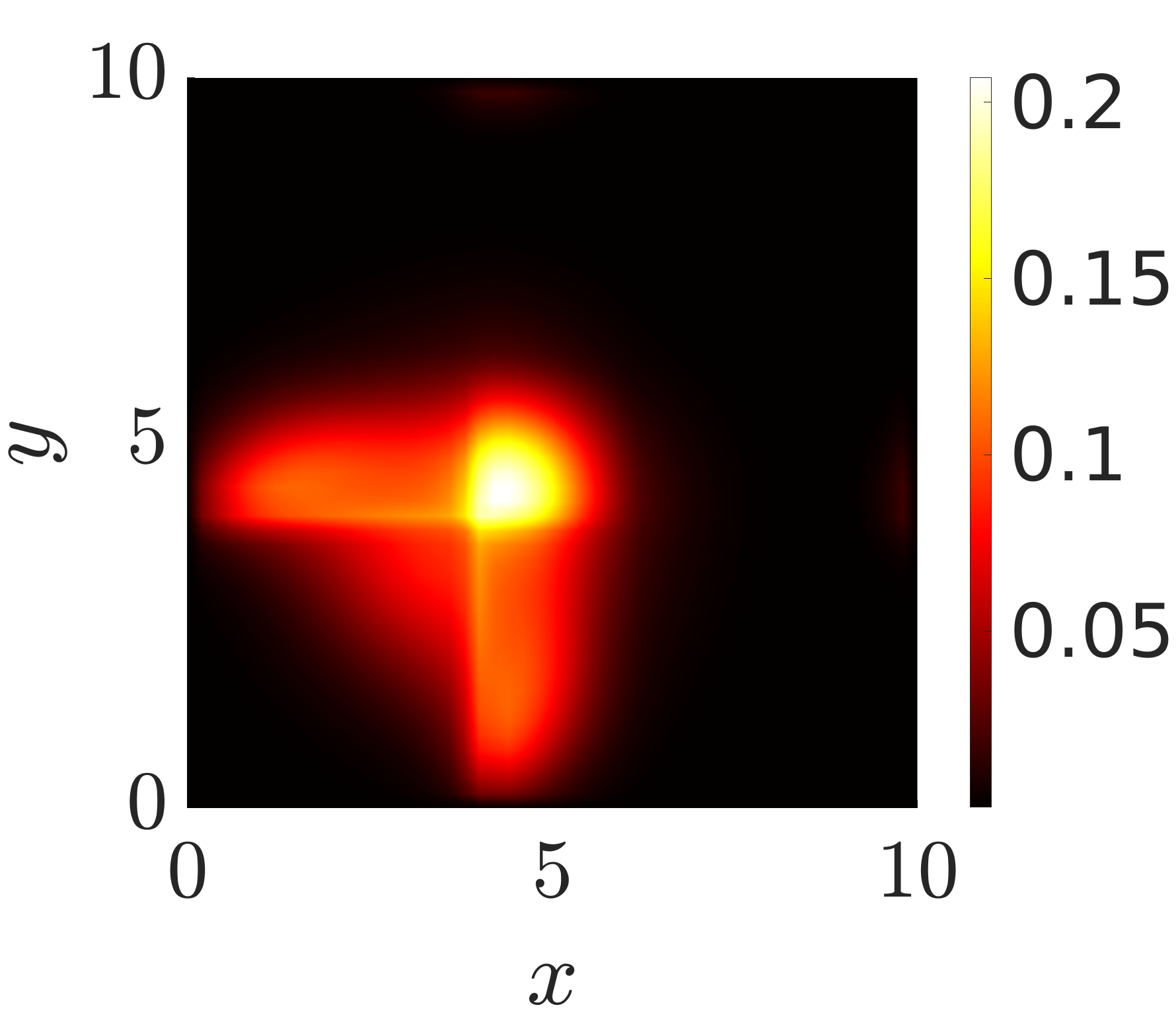}}
\caption{{\bf Test 3}. Dynamics for $q$ given by the cross structure in Fig. \ref{cross} (left). (a)  Initial macroscopic density (centered at $(5,5)$) for the simulations represented in (b,c); (d) Initial macroscopic density (centered at $(4,4)$) for simulations represented in (e,f). $q$ is given by \eqref{q.cross}, see text for details, while $\mu$ has a similar structure, but for (b,e) $\theta_{\mu} = \theta_{q}$ and for (c,f) $\theta_\mu=\theta_q^{\bot}$.}
\label{fig.4}
\end{figure}

\section{Application: movement on fabricated anisotropic surfaces}\label{s:Doyle}

As an application-oriented investigation we return to the cell migration studies of Doyle {\em et al} \cite{doyle2009}, illustrated in Figure \ref{doyleschema}. These experiments rely on a ``photopatterning'' technique, allowing fabrication of imprinted fibronectin micro-structures. Laying down parallel aligned stripes imitates aligned fibres and, confronted by such environments, migratory cells (fibroblasts and keratinocytes) orient accordingly, extending protrusions and forming the adhesive attachments that allows movement along the alignment axis. Significantly, highly aligned environments lead to 
a substantial increase in velocity, for example two-fold (for fibroblasts) or three-fold (for keratinocytes) 
over corresponding movements on an unaligned surface.

While the authors of \cite{doyle2009} do not explicitly measure the turning rate $\mu(x,w)$, they do measure net velocity and persistence. We note that the emphasis of the experiments in \cite{doyle2009} was on cell orientation and morphology, not so much turning rates, so available data remains sparse. Subsequently, rather than a detailed attempt of model fitting, we perform a qualitative comparison to explore how direction-dependent turning rates will impact on the movement patterns of cells on different surfaces.


We specifically focus on the ``transition'' experiment, schematically illustrated in Figure \ref{doyleschema}. Here quasi-1D regions were interrupted by two isotropic 2D forms: {\em Case A}, a {\em uniformly} isotropic region, or {\em Case B}, a region of perpendicular and criss-crossing stripes. As indicated earlier, cells that enter the isotropic region from a neighbouring quasi-1D region round-up and extend protrusions in multiple directions. For {\em Case A} the cell's net movement is dramatically reduced, losing its direction and subsequently performing what appears as an unbiased random walk. In {\em Case B} movement is also arrested and reorientation occurs, but there can be a subsequent significant movement along one of the two perpendicular directions, followed by further reorientations. Overall, the translocations of the cell in {\em Case B} seem to be significantly longer. Here we will show that a direction-dependent turning rate can generate this distinct behaviour.

We adopt a two-pronged approach for the analysis, computing first the macroscopic diffusion tensor for cells migrating in the completely unaligned tissue of case A or the criss-cross pattern of case B. While this is a macroscopic-level analysis (and the underlying experiments are mesoscopic) it will provide valuable evidence of variation in critical movement characteristics according to cell orientation/turning behaviour. We then provide simulations of the mesoscopic-level transport equation, indicating whether the model can indeed recapitulate the observations. Note that for convenience we assume an {\em a priori} rescaling that fixes the cell speed $s$, i.e. $V=s\mathbb{S}^1$.

\subsection{Control Case} 

As a control consider a constant turning rate $\mu$ and, in turn, a constant mean travel time $\tau$. Here the macroscopic diffusion tensor is computed from formula (\ref{diff.1}) as 
\[ \mathbb{D} = \frac{1}{\mu} \mathbb{V}_q.\]
Under {\em Case A} the middle region is completely non-oriented, hence $q=\frac{1}{2\pi}$ (the uniform distribution). Then, 
\begin{equation}\label{above}
\mathbb{V}_q = \int_V w\otimes w \frac{1}{2\pi} dw = \frac{2\pi}{2} \mathbb{I}\frac{1}{2\pi} = \frac{1}{2} \mathbb{I}, 
\end{equation}
where $\mathbb{I}$ denotes the identity matrix. The criss-cross configuration of {\em Case B} can be described by combining two bi-modal von-Mises distributions, in perpendicular directions $e_1 = (1,0)$ and $e_2 = (0,1)$ but with the same concentration parameter $k$:
\begin{equation}\label{bvm12}
q(w) = \frac{1}{2}\left( \bvmone + \bvmtwo\right).
\end{equation}
Following standard calculations (e.g. see \cite{HMPS}), we obtain 
\[ \mathbb{V}_q = \frac{1}{2}\left( 1 -\frac{I_2(k)}{I_0(k)}\right) \mathbb{I} + \frac{I_2(k)}{I_0(k)} \frac{1}{2}\left(e_1 e_1^T + e_2 e_2^T\right) .\]
Now,
\[ e_1 e_1^T = \left( \begin{array}{cc} 1 & 0 \\ 0 & 0 \end{array}\right), \qquad \mbox{and} \quad e_2 e_2^T = \left( \begin{array}{cc} 0 & 0 \\ 0 & 1 \end{array}\right),\]
so
\begin{equation}\label{Vperpenticular}
\mathbb{V}_q = \frac{1}{2}\mathbb{I} -\frac{I_2(k)}{2I_0(k)} \mathbb{I} + \frac{I_2(k)}{2I_0(k)} \mathbb{I}  = \frac{1}{2} \mathbb{I},
\end{equation}
which coincides exactly with the calculation for {\em Case A}, i.e. (\ref{above}). Therefore, under a constant turning rate there should be no macroscopic difference between Case A and Case B, even if cells bias their orientation along the
criss-crossed fibres.

\subsection{Direction-dependent turning}

We now consider $w$-dependence in the turning rate, i.e. $\mu(w)$. Specifically, we choose $\mu$ such that the rate of turning is reduced if cells are migrating along the direction of dominating alignment. For the analysis we focus only on the central region, so both {\em Case A} and {\em Case B} can be regarded as spatially homogeneous (i.e. not depending on $x$) and we can use equation (\ref{macroDp2}):
\[ \mathbb{D} = \int w\otimes w \frac{T}{\mu} dw, \qquad 
T = \frac{q}{C \mu}, \qquad C(x) = \int_V \frac{q}{\mu} dw. \]
We again choose $q$ to be a combination of bimodal von-Mises distributions in the two perpendicular directions $e_1$ and $e_2$, as given in (\ref{bvm12}). However, now we assume that $\mu\sim q^{-1}$, so that
\[
\mu=\frac{1}{2\pi} \dfrac{2}{\left({\bvmtwo}+\bvmone\right)}\,.
\]
where we chose the normalization constant to be $(2\pi)^{-1}$ such that in the isotropic limit of $k\to 0$ we obtain
\[ \lim_{k\to 0} \mu = \frac{1}{2\pi} \frac{2}{\frac{1}{2\pi}+\frac{1}{2\pi} } = 1.\]
Then 
\begin{equation}\label{thirdposer}
\frac{T}{\mu} = \frac{\pi^2}{2 C } \left({\bvmtwo}+\bvmone\right)^3.
\end{equation}

For this choice of $\mu$ and $q$ we can compute the normalization constant $C$ as 
\begin{eqnarray}
    C(x) &=& \int_V\frac{q}{\mu} dw  = \int_V\frac{\pi}{2} \left(\bvmtwo+\bvmone\right)^2 dw \nonumber \\
    &=& \frac{\pi}{2}\frac{1}{(4\pi I_0(k))^2} \int_V 4\pi I_0(2k) (\bvmm_{2k, e_1} +\bvmm_{2k, e_2}) + 4 \nonumber \\
    && \hspace*{2.6cm} + 8\pi I_0(\sqrt{2}k) \left(\bvmm_{\sqrt{2} k, \frac{e_1+e_2}{\sqrt{2}}} + \bvmm_{\sqrt{2}k, \frac{e_1-e_2}{\sqrt{2}}}\right) dw \nonumber \\ 
    &=&  \frac{1}{4 I_0(k)^2 }\bigl(I_0(2k) + 1 + 2 I_0(\sqrt{2} k) \bigr).\label{normalizationlonstant}
    \end{eqnarray}
We note the isotropic limit $\lim_{k\to 0} C(x) = 1$, since $I_0(0)=1$.

Next we compute the third power in (\ref{thirdposer})
\begin{eqnarray*}
\frac{T}{\mu} &=& \frac{\pi^2}{2 C}  \left( \frac{I_0(3k)}{I_0(k)} \Bigl[\bvmm_{3k,e_1}+\bvmm_{3k,e_2}\Bigr] + 3\Bigl[\bvmone+\bvmtwo\Bigr]\right.\\
&& \hspace*{1.6cm} + 3\Bigl[ \bvmm_{k, 2e_1+e_2} +\bvmm_{k,2e_1-e_2} + \bvmm_{k, e_1+2e_2} + \bvmm_{k,e_1-2e_2} \\
&& \hspace*{2.6cm} + 2(\bvmone + \bvmtwo)\Bigr]\Bigr)
\end{eqnarray*}
Notably, vectors $2e_2+ e_1$ {\em etc} are not unit vectors so we rescale:
 \[  \zeta_1=\frac{1}{\sqrt{5}}(2,1)^T\,,\qquad  \zeta_2=\frac{1}{\sqrt{5}}(1,-2)^T\,,\qquad  \xi_1=\frac{1}{\sqrt{5}}(2,-1)^T\,,\qquad  \xi_2=\frac{1}{\sqrt{5}}(1,2)^T.\]
For the example $2e_1+e_2$ this gives
 \[ \bvmm_{k,2e_1+e_2} = \frac{I_0(\sqrt{5} k)}{I_0(k)} \bvmm_{\sqrt{5}k, \zeta_1}, \]
and similar for the other terms. With unit vectors everywhere which, moreover, are pairwise perpendicular ($\zeta_1\cdot \zeta_2 =0, \, \xi_1\cdot\xi_2=0$) we have
\begin{eqnarray*}
\frac{T}{\mu} &=& \frac{\pi^2}{2 C}  \left( \frac{I_0(3k)}{I_0(k)} \Bigl[\bvmm_{3k,e_1}+\bvmm_{3k,e_2}\Bigr] + 3\Bigl[\bvmone+\bvmtwo\Bigr]\right.\\
&& \hspace*{1.5cm} + 9 (\bvmone+\bvmtwo)
+ 3\frac{I_0(\sqrt{5} k)}{I_0(k)} \left(\bvmm_{\sqrt{5}k, \zeta_1} +\bvmm_{\sqrt{5}k, \zeta_2}\right) \\
&&\hspace*{1.5cm} + 3\frac{I_0(\sqrt{5} k)}{I_0(k)} \left(\bvmm_{\sqrt{5}k, \xi_1} +\bvmm_{\sqrt{5}k, \xi_2}\right)\Bigr).
\end{eqnarray*}
Noting that the second moment of bimodal von Mises distributions with pairwise perpendicular unit vectors is $\frac{1}{2}\mathbb{I}$ (see (\ref{Vperpenticular})), we find a diffusion tensor 
\[ \mathbb{D} = d(k) \mathbb{I}, \qquad d(k) = \frac{\pi^2}{2 C (4\pi I_0(k))^2}   \left( \frac{I_0(3k)}{I_0(k)} + 6 \frac{I_0(\sqrt 5 k)}{I_0(k)} + 9\right).\]
Substituting the normalisation constant $C$ from (\ref{normalizationlonstant}) we obtain
\[d(k) =  \frac{I_0(3k) + 9 I_0(k) + 6 I_0(\sqrt{5}k) } {8 I_0(k) ( I_0(2k) + 2I_0(\sqrt{2} k) + 1)}. \]
Again we consider the isotropic limit 
\[ \lim_{k\to 0} d(k) = \frac{1 + 9 + 6}{ 8(1+2+1) } = \frac{1}{2}, \]
which has the correct scaling as for the isotropic case. The diffusion coefficient $d(k)$ is plotted in Figure \ref{fig:dofk}, where we observe that for small $k$ there is negligible change but for $k>3$ we see clear and sustained increase for the diffusion coefficient. 
\begin{figure}
    \centering
 \centerline{\includegraphics[width=4.5cm]{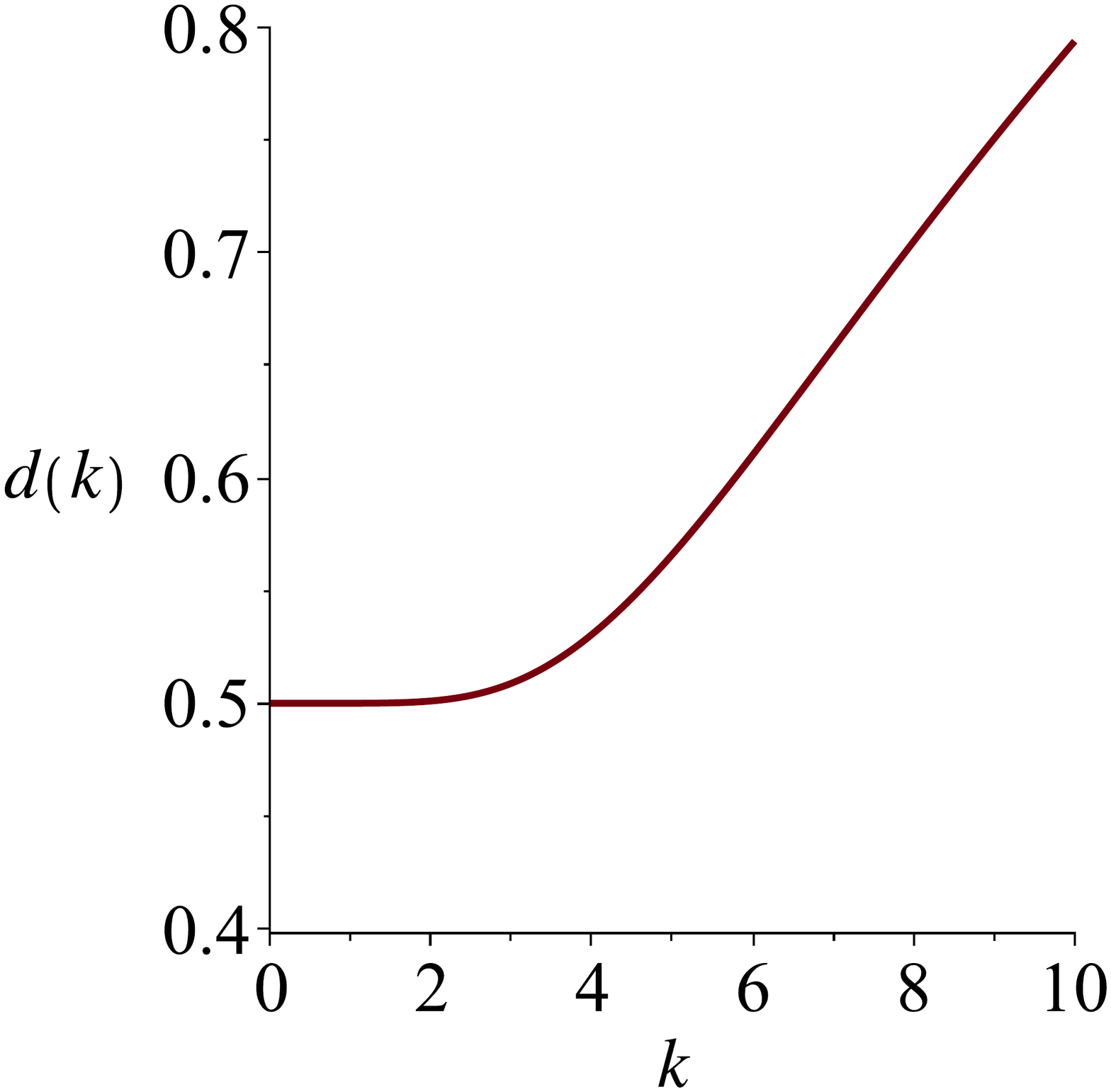}\hspace*{0.7cm}\includegraphics[width=4.5cm]{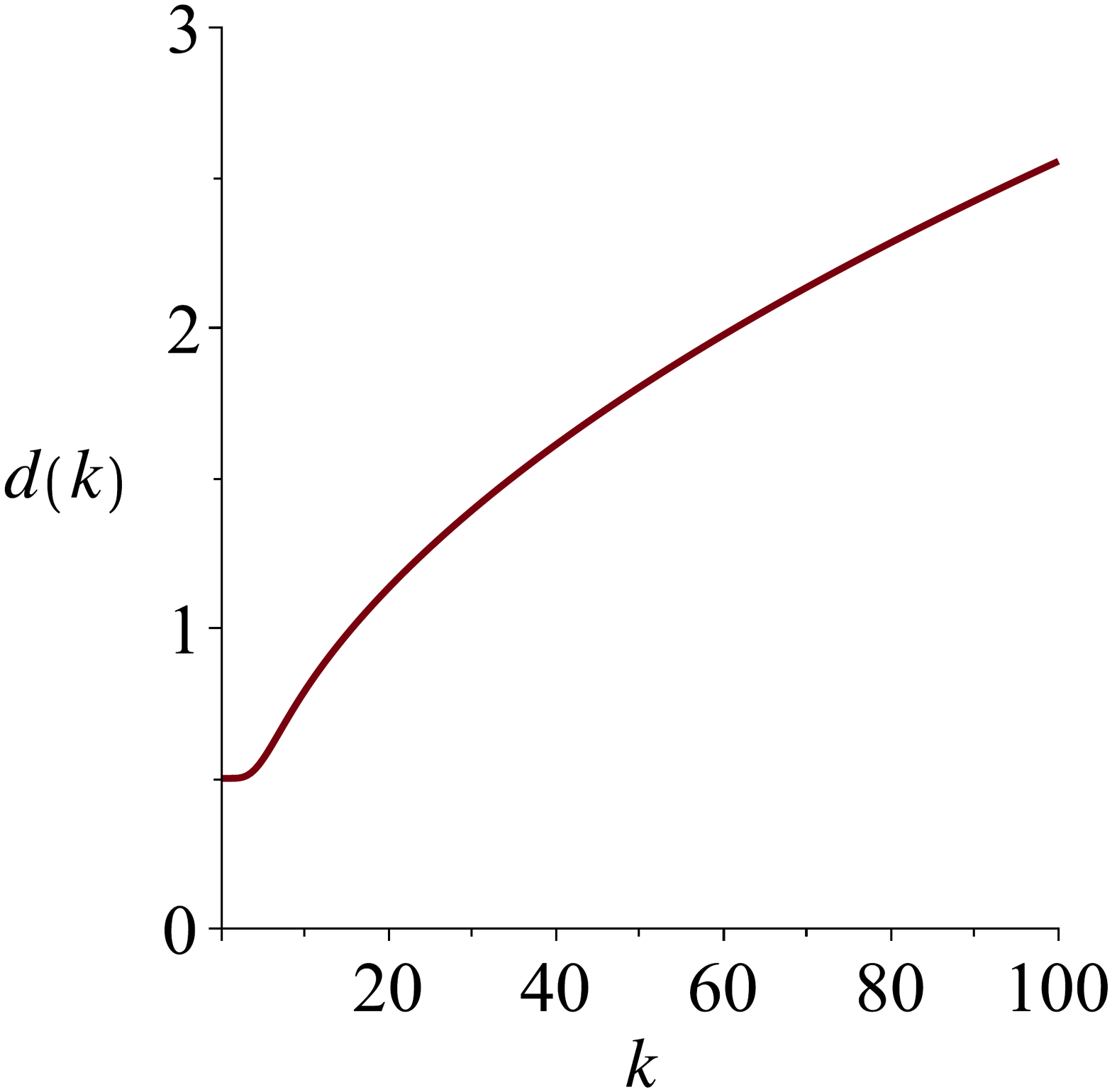}}
    \caption{Diffusion coefficient $d(k)$ as function of $k$ for $k\in[0,10]$ on the left and $k\in[0,100]$ on the right. The diffusion coefficient grows very slowly for small $k$, but then grows more rapidly for $k \gtrsim 3$. For large $k$ the increase is approximately $\propto \sqrt{k}$.} 
    \label{fig:dofk}
\end{figure}

\subsection{Transport model simulations}

The above analysis indicates that a direction-dependent turning rate coupled to a criss-cross fibre network substantially increases the diffusion, compared to either a constant turning rate or a completely isotropic network. To simulate this situation, we consider a rectangular domain $\Omega=[0,15]\times[0,5]$ with $V=\mathbb{S}^1$ ($s=1$) and define local environments corresponding to those illustrated in Figure \ref{doyleschema}. 
Specifically, to describe the parallel alignment in the left and right regions, we set $q(x,w)=e^{k w\cdot \theta}/2\pi I_0(k)$
with $k=25$ and $\mu(x,w)=1$ when $x<5$ or $x>10$. To replicate the completely isotropic scenario of \textit{Case A} we define the central region ($5\le x\le 10$) by $q(x,w)=1/2\pi$ while for the criss-cross network of \textit{Case B} we choose  $q(x,w)=\left(\bvmone+\bvmtwo\right)/2$ in the central region and take the anisotropy constant $k$ to be a variable model parameter. Consequently, for \textit{Case A} $\mu(w)=1$ for the full domain, while in \textit{Case B} we chose  $\mu(x,w)=1/(2\pi q)$ if $5\le x\le 10$. We initialise the density distribution as uniformly distributed in $\mathbb{S}^1$ with the initial macroscopic density ($\bar{p}_0$) as defined in \ref{Test1}(a), but centered at the coordinate $(3.5,2.5)$, i.e. inside the left-most region of highly aligned fibres. Simulations of the corresponding transport equation (\ref{transport.general})-\eqref{operator.2} are shown in Figure \ref{fig:Doyle}. The first row illustrates results for Case A, while subsequent rows are for Case B with increasing values $k = 0,3,10,25$, respectively. 

Under Case A (Figure \ref{fig:Doyle} first row), as cells reach the isotropic region we observe a gradual diffusive-like spread, consistent with the earlier analysis. Under Case B, but setting $k=0$ (Figure \ref{fig:Doyle} second row), generates equivalent behaviour: in line with the prediction that isotropic criss-cross networks do not alter the macroscopic dynamics {\em when the turning rate is constant}. For $k>0$, however, we observe distinct behaviour. This is minimal for small $k$ (e.g. $k=3$, Figure \ref{fig:Doyle} third row) but becomes clear for large $k$, e.g. $k=10$ (fourth row) and $k=25$ (fifth row), respectively. A noteworthy phenomenon lies in the ``droplet'' detaching from the main swarm for high anisotropy parameter values ($k=10, 25$). Here, a fraction of invaders maintain the left to right direction on reaching the central region, detaching from the main swarm. This observation is unexpected and could be of interest to confirm experimentally.

\begin{figure}[!htbp]
    \centering
    \includegraphics[scale=0.1]{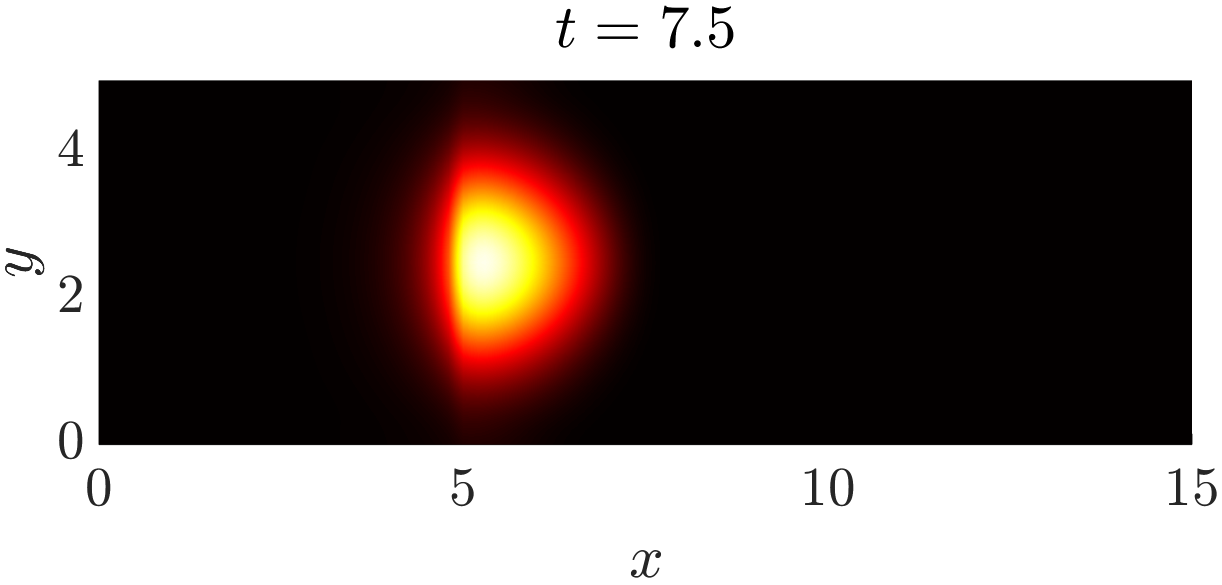}
    \includegraphics[scale=0.1]{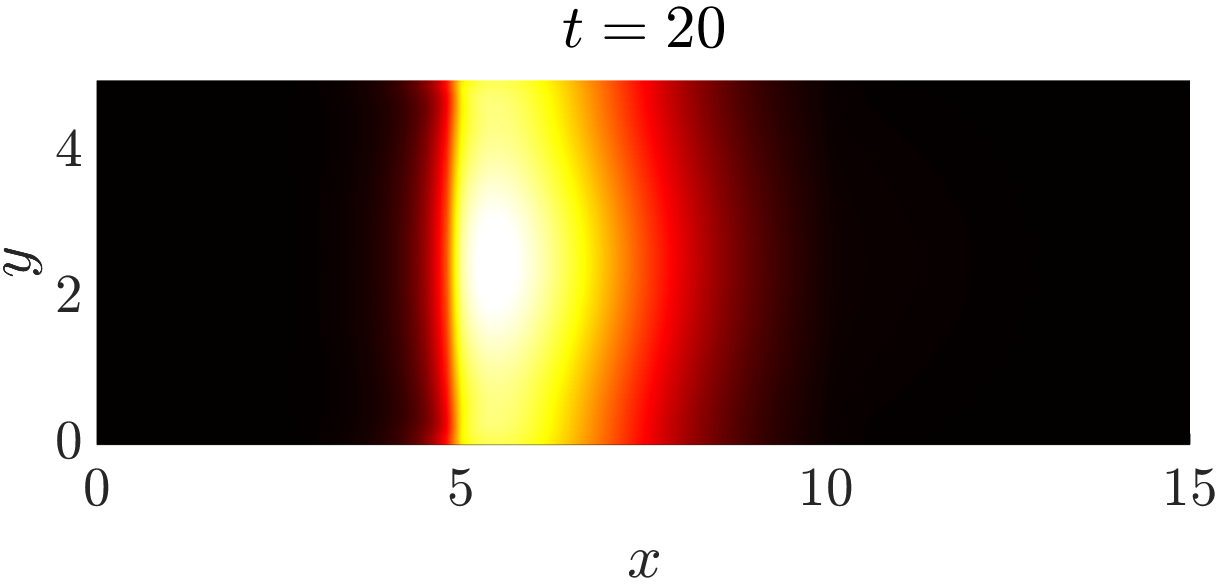}
    \includegraphics[scale=0.1]{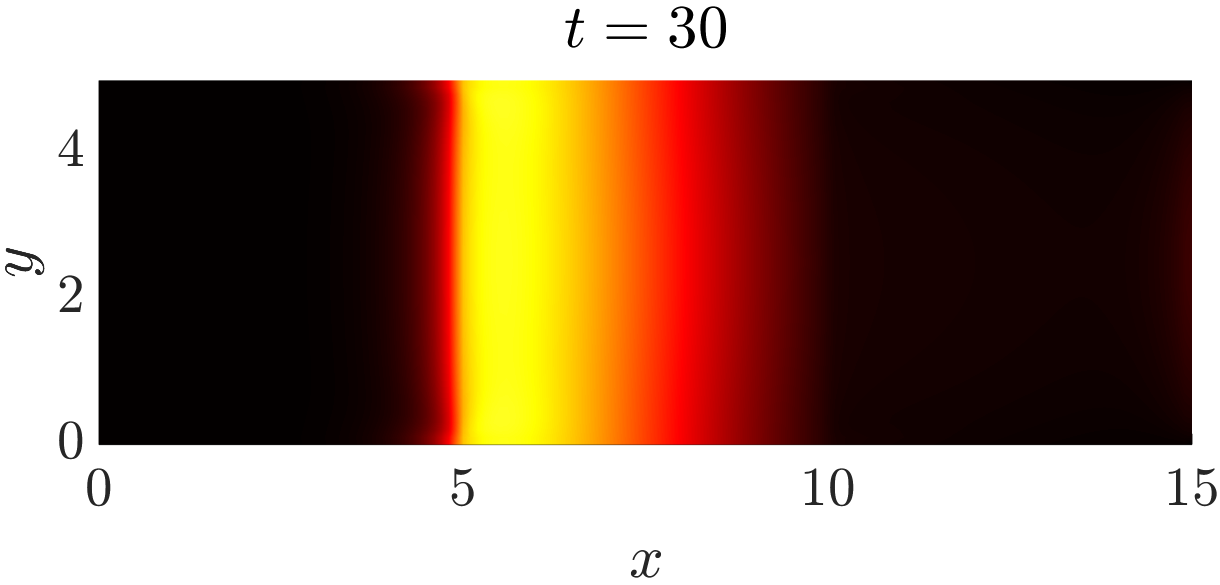}\\
    \includegraphics[scale=0.1]{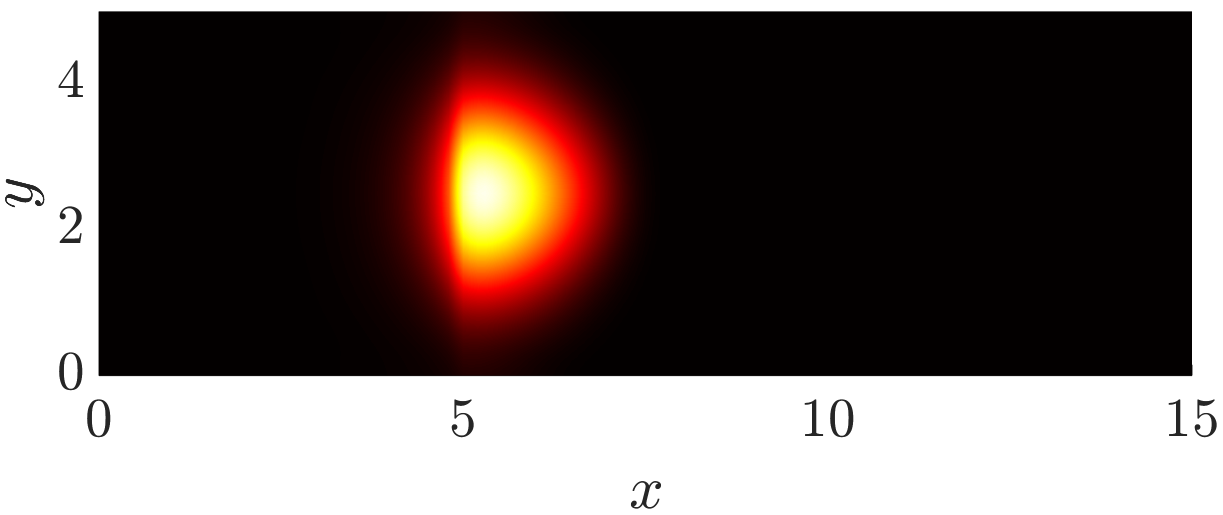}
   \includegraphics[scale=0.1]{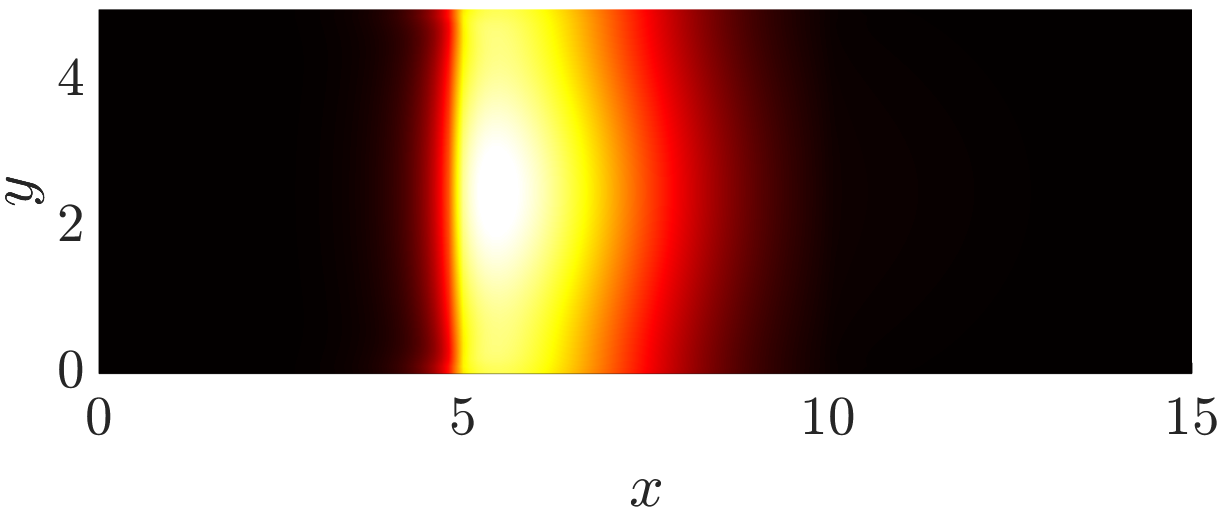}
    \includegraphics[scale=0.1]{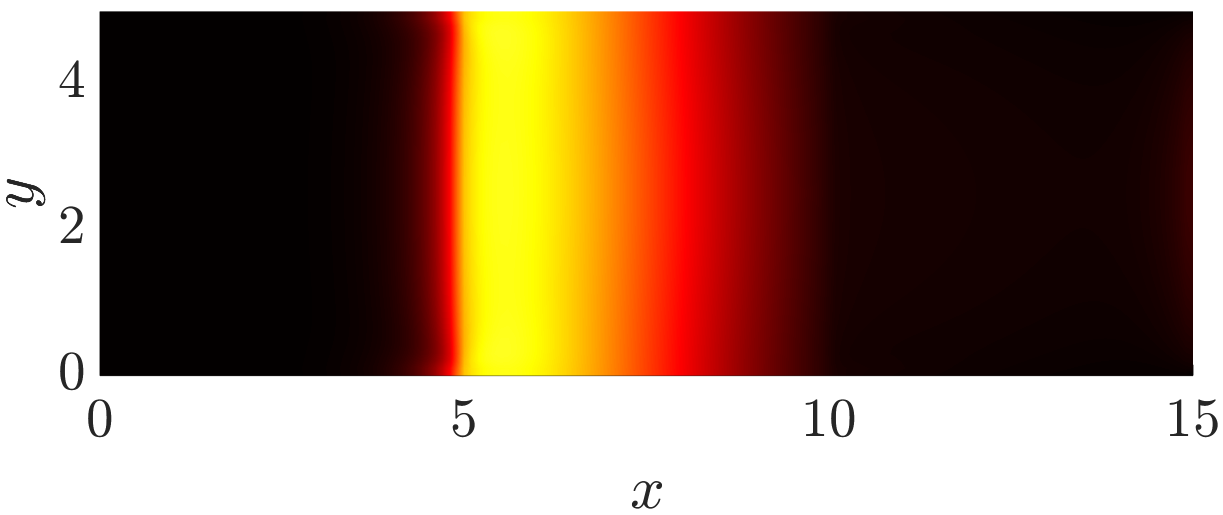}\\
    \includegraphics[scale=0.1]{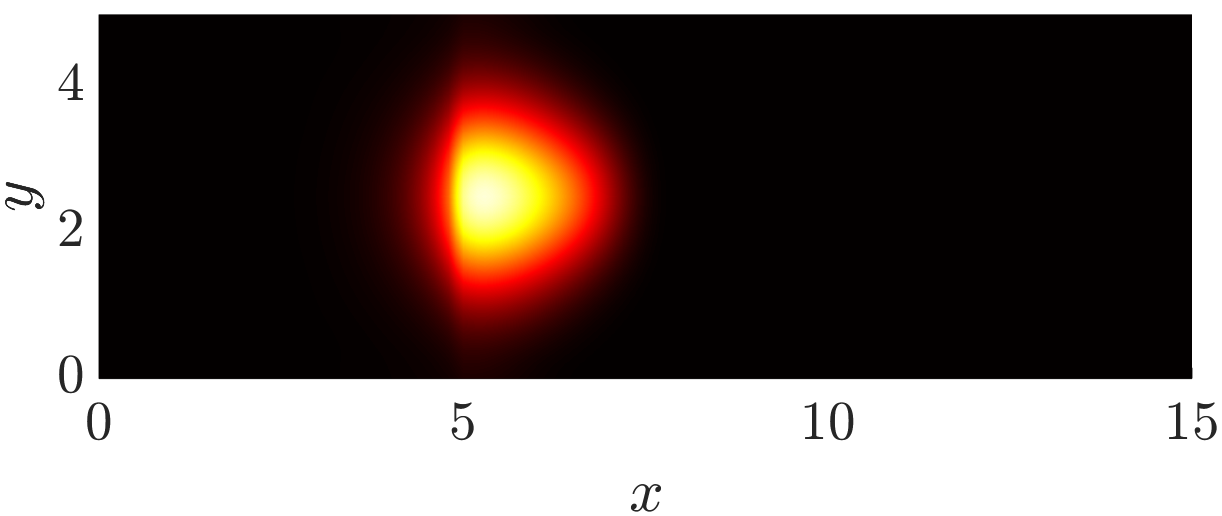}
    \includegraphics[scale=0.1]{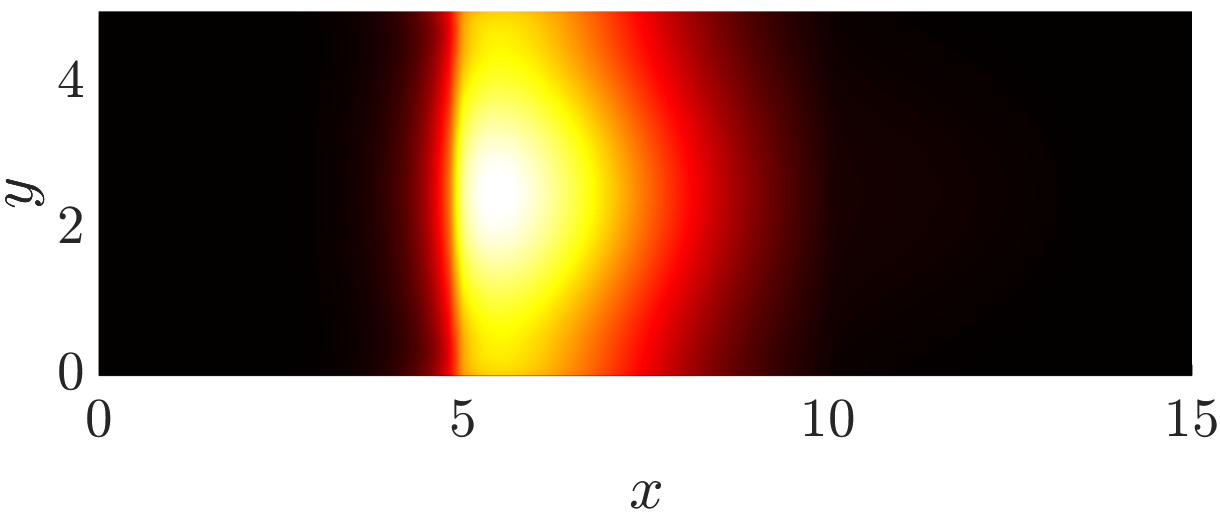}
    \includegraphics[scale=0.1]{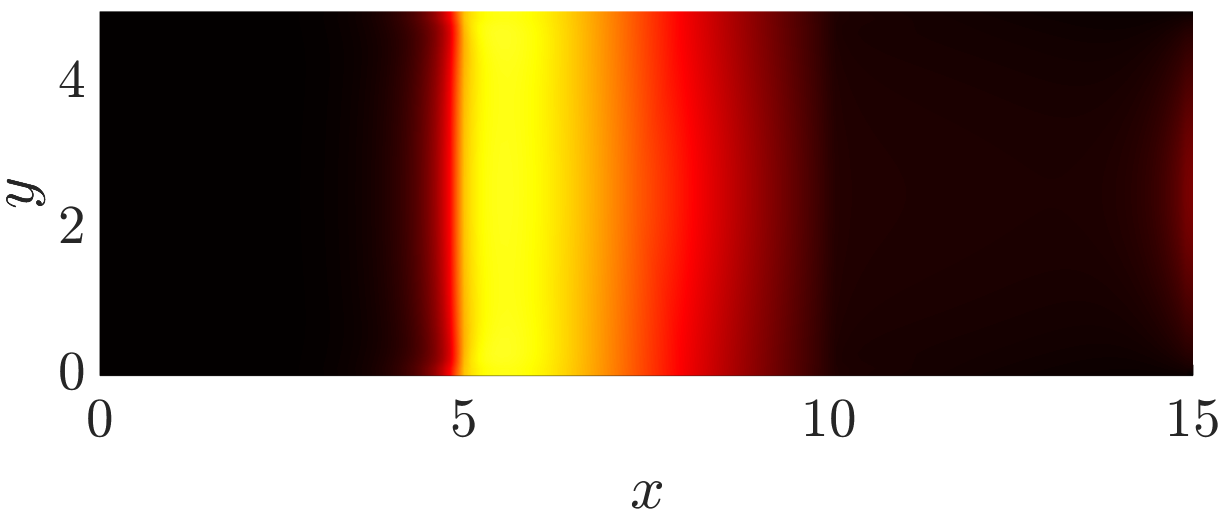}\\
    \includegraphics[scale=0.1]{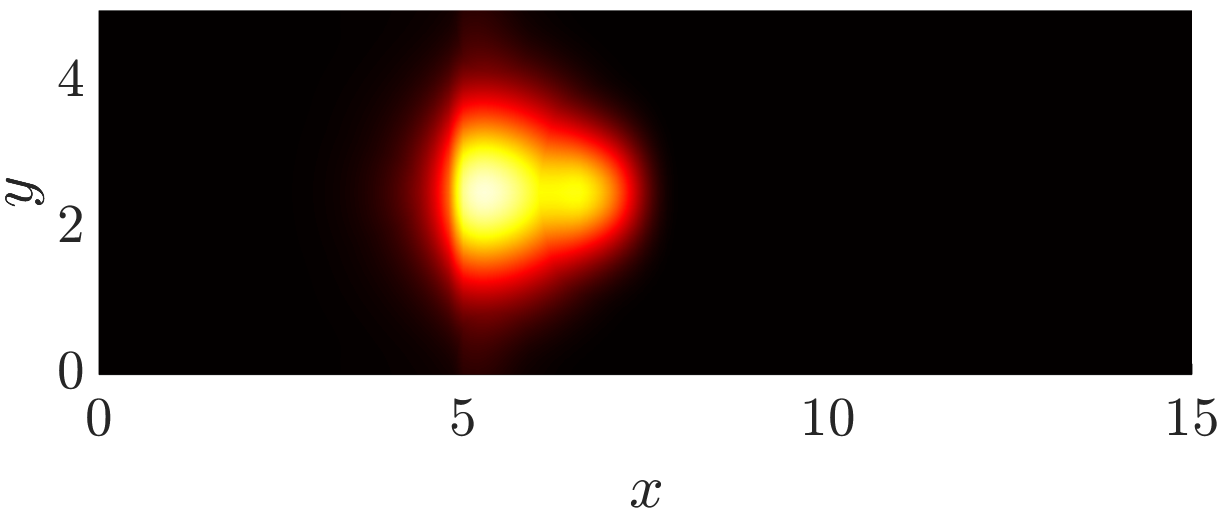}
    \includegraphics[scale=0.1]{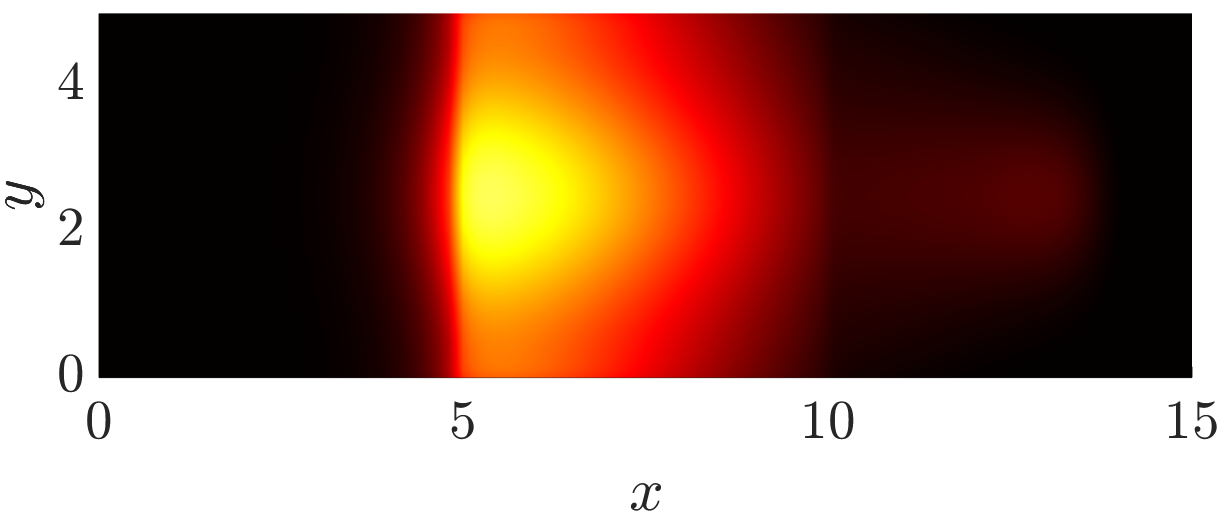}
    \includegraphics[scale=0.1]{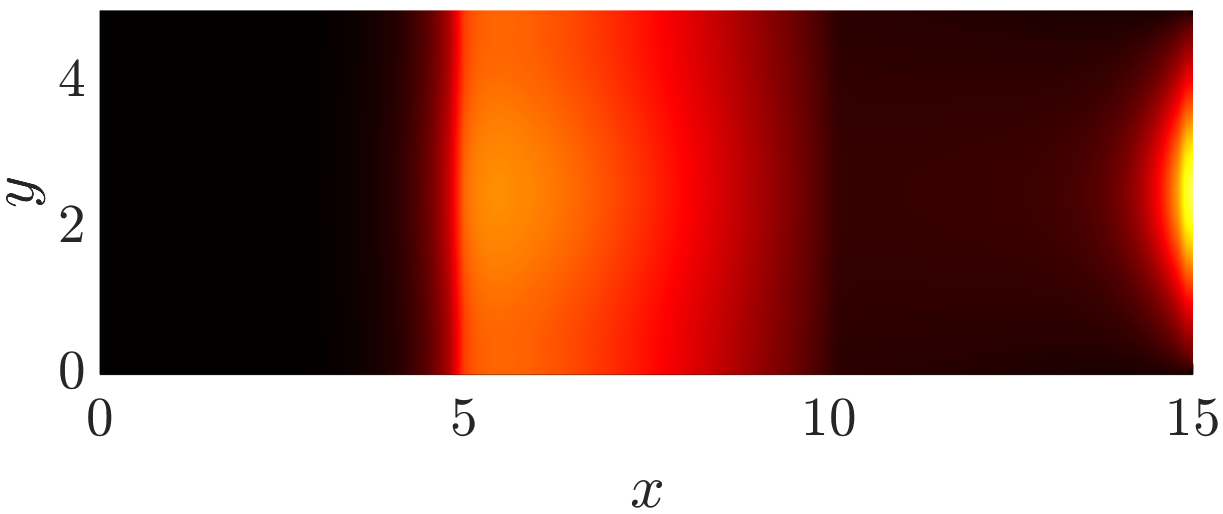}\\
    \includegraphics[scale=0.098]{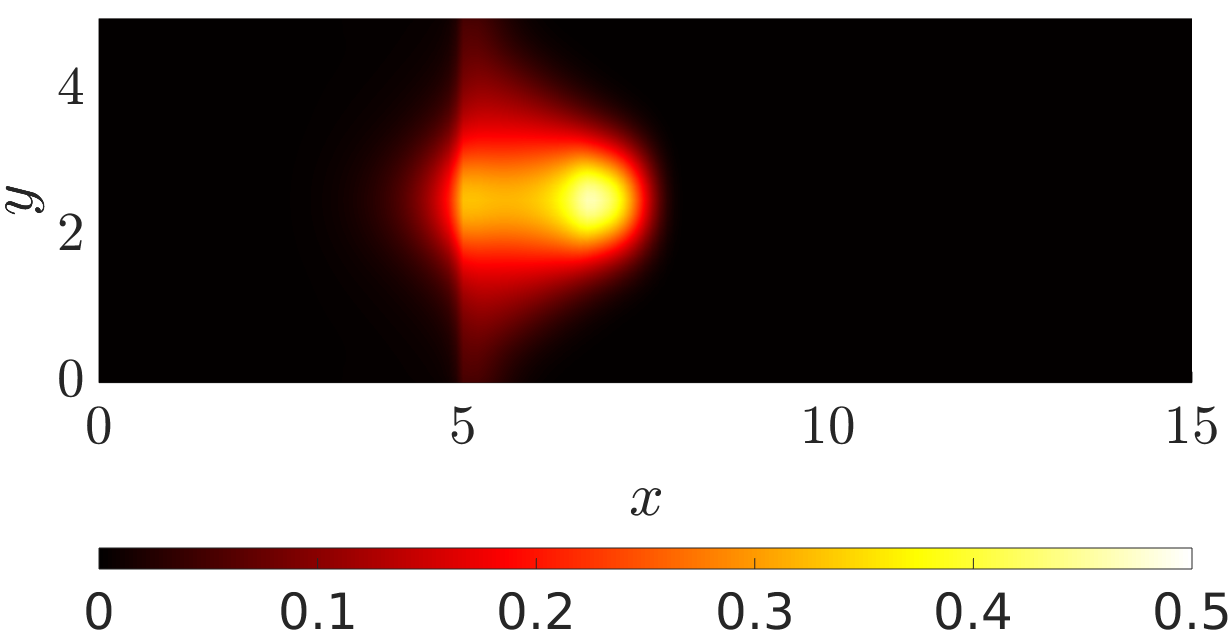}
    \includegraphics[scale=0.098]{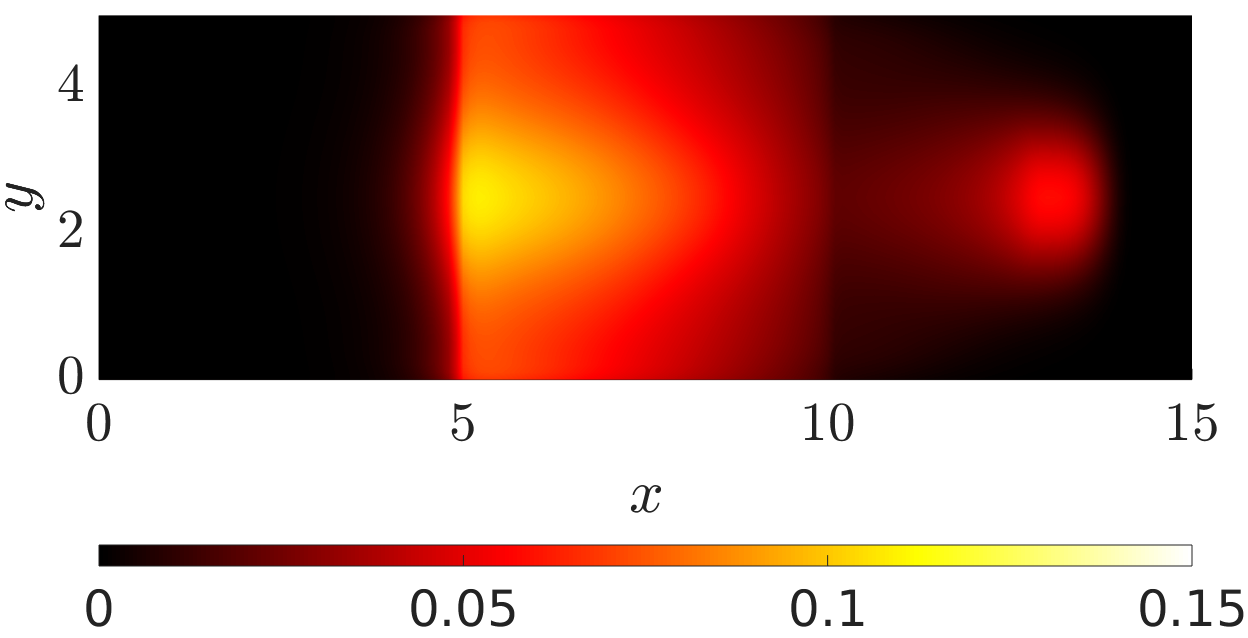}
   \includegraphics[scale=0.098]{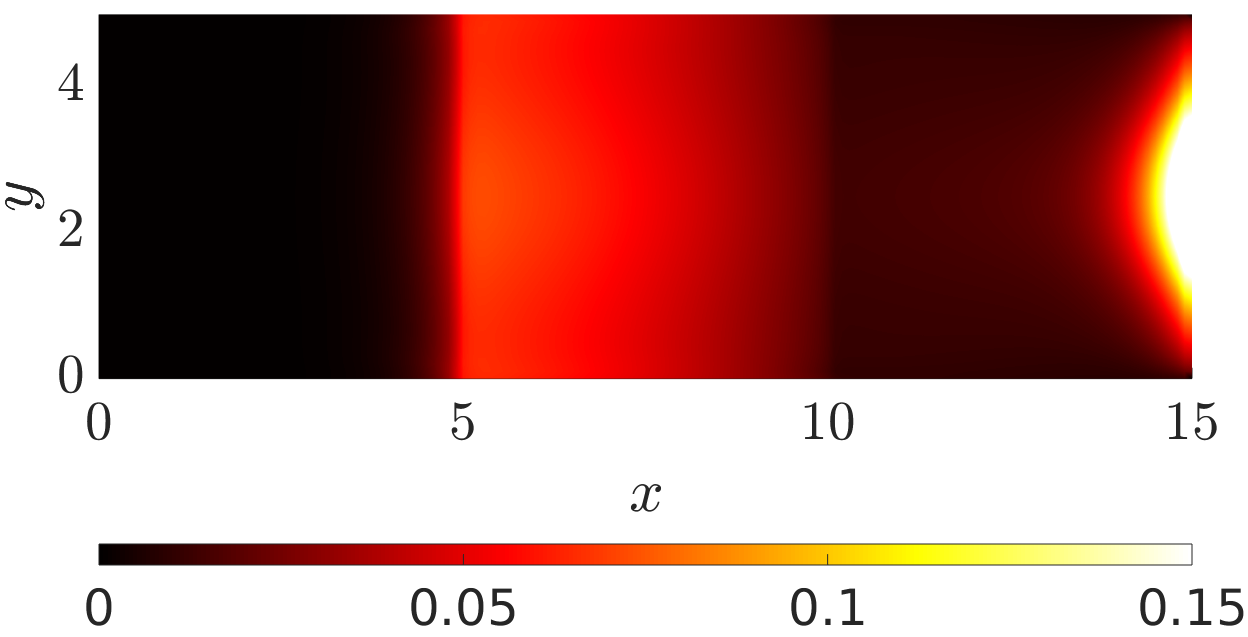}
    \caption{Cell movement on fabricated anisotropic surfaces. For all simulations we set $q(x,w)=e^{k w\cdot \theta}/(2\pi I_0(k))$ with $k=25$ and $\mu(x,w)=1$ for $x<5$ and $x>10$. First row: \textit{Case A} with $q(x,w)=1/(2\pi)$ if $5\le x\le 10$ and $\mu(w)=1$ on the full domain. Second to final row: $q(x,w)=\left(\bvmone+\bvmtwo\right)/2$ and $\mu(x,w)=1/(2\pi q)$ if $5\le x\le 10$, so that $s=1$. In particular: (second row) $k=0$, (third row) $k=3$, (fourth row) $k=10$, (final row) $k=25$.}
    \label{fig:Doyle}
\end{figure}

\section{Conclusions}\label{s:discussion}

In this paper we have analysed a transport equation for cell migration along oriented fibres, extending the model proposed in \cite{Hillen.05} to include a turning rate that depends on the microscopic velocity of the cells and, in turn, on the anisotropic structure of their environment. This key extension admits a more nuanced and realistic description for how a migratory cell population responds to alignment of the environment, with several recent studies investigating the impact of oriented collagen fibre networks on the orientation, speed and persistence of movement of cells (e.g. see \cite{riching2014,Ray2017,ray2018}).

Formally, the resulting equation is a non-homogeneous linear Boltzmann equation with a micro-reversible process in which the cross-section is factorised into the distribution of the fibres and the direction-dependent turning rate. The dependence of the turning rate on the orientation alters the entire mathematical set up with respect to \cite{Hillen.05}. First, the equilibrium distributions of the system now depend on both the turning rate and on the transition probability. Consequently, the null-space of the turning operator needs to be equipped with a direction-dependent integration weight. Moreover, the average of the equilibrium state, defined in Eq. \eqref{mean.dir.fr}, becomes linked to both the fibre distribution and the turning frequency. This implies new solvability conditions for the parabolic and hyperbolic limits. We study many special cases, some of which are relevant for applications while others provide a theoretical connection to previous results. 

Further, the orientation-dependent turning rate permits the definition of a new adjoint persistence \eqref{adj.per.gen}, taking into account the cell persistence encoded in the turning rate. Consequently, direction-dependent turning rates are capable of generating a persistent random walk even under a symmetric configuration of fibres. In this framework, there is no directional persistence if \eqref{mean.dir.fr} vanishes. This holds true, for example, if Eq. \eqref{nec.diff.cond.} holds true, which means that fibres are bi-directed and cells having a certain direction have the same turning frequency regardless of the sense in which they are travelling on that direction.

To illustrate the broader relevance of the extended framework, we considered an application to cell migration across precisely engineered network arrangements, for example as fabricated in the studies of \cite{doyle2009}. Under the the original framework of \cite{Hillen.05} the two 
configurations shown in Figure \ref{doyleschema} generate equivalent behaviour, yet extending to direction-dependent turning rates could result in a markedly different response. Specifically, under the criss-cross network a markedly faster passage could be observed which, translated to the macroscopic level, yielded enhanced diffusion. Clearly, such behaviour has potential to significantly alter the predictions from modelling invasion pathways in complex anisotropic tissues, for example glioma invasion in the central nervous tissue \cite{Painter2013,Engwer_Hillen_Surulescu.15,Swan.17}.

Within the current work we have restricted to movement under negligible modification of the network, as in amoeboid movement or cell migration on engineered fibronectin strips. Under {\em in vivo} mesenchymal migration, however, contact-guided movement can be coupled to significant matrix remodelling, for example fibres becoming aligned along the migratory path; simulations of the simpler transport model in this scenario revealed symmetry breaking behaviour, with cells forming and migrating along a network of aligned ``cellular highways'' \cite{Hillen.05,Painter.09}. The addition of direction-dependent turning has clear potential to impact on such pattern formation scenarios and a key aim of future studies will be to investigate this phenomenon.\\

\noindent{\bf Acknowledgements:} {\small This work was stimulated by the PhD thesis of Amanda Swan \cite{AmandaThesis}, who started looking into the parabolic scaling of a velocity dependent turning rate. TH is grateful to support through the Natural Sciences and Engineering Research Council of Canada. NL is recipient of a Post-Doc grant of the Italian National Institute of High Mathematics (INdAM) and ackowledges support by the Italian Ministry for Education,
University and Research (MIUR) through the ``Dipartimenti di Eccellenza'' Programme (2018-
2022), Department of Mathematical Sciences, G. L. Lagrange, Politecnico di Torino (CUP:
E11G18000350001), and support from ``Compagnia di San Paolo'' (Torino, Italy).  }

\bibliography{references}

\label{lastpage}

\end{document}